\documentclass{article}
\usepackage{algorithm2e}
\usepackage{boxedminipage}
\usepackage{amsmath, amsthm, amssymb, bbm, bm, esvect}
\usepackage{graphicx}
\usepackage{verbatim}
\usepackage{natbib}
\usepackage{caption}
\usepackage{subcaption}
\usepackage{fancyvrb}
\usepackage{enumerate}
\usepackage{relsize}
\usepackage{mathtools}
\usepackage{xcolor}
\usepackage{soul}
\usepackage{esvect}
\usepackage{dirtytalk}
\usepackage{float}

\usepackage{tikz}
\usetikzlibrary{fit,positioning}

\usepackage{hyperref}
\usepackage[margin=1.5in]{geometry}
\usepackage{chngpage}
\hypersetup{colorlinks,citecolor=blue,urlcolor=blue,linkcolor=blue}

\usepackage{stefan_tex}
\graphicspath{{./plots/}}

\theoremstyle{plain}
\newtheorem{prop}{Proposition}

\newtheorem{coro}[prop]{Corollary}
\newtheorem{lemm}[prop]{Lemma}
\newtheorem{theo}[prop]{Theorem}

\theoremstyle{definition}
\newtheorem{exam}{Example}
\newtheorem{defi}[exam]{Definition}
\newtheorem{assu}{Assumption}

\theoremstyle{remark}

\newtheorem{rema}[prop]{Remark}

\date{Draft version \ifcase\month\or
	January\or February\or March\or April\or May\or June\or
	July\or August\or September\or October\or November\or December\fi \ \number%
	\year\ \  }
\author{Xinkun Nie\\ \texttt{xinkun@stanford.edu}
	\and Emma Brunskill \\ \texttt{ebrun@cs.stanford.edu}
	\and Stefan Wager \\ \texttt{swager@stanford.edu}}

\title{Learning When-to-Treat Policies}
\begin{document}

\maketitle

\theoremstyle{definition}
\newtheorem{deficmp}{Definition}
\renewcommand\thedeficmp{\arabic{deficmp}b}
\newtheorem{assucmp}{Assumption}
\renewcommand\theassucmp{\arabic{assucmp}b}

\begin{abstract}
Many applied decision-making problems have a dynamic component: The policymaker needs not only to choose whom to treat, but also when to start which treatment. For example, a medical doctor may choose between postponing treatment (watchful waiting) and prescribing one of several available treatments during the many visits from a patient. We develop an ``advantage doubly robust'' estimator for learning such dynamic treatment rules using observational data under the assumption of sequential ignorability. We prove welfare regret bounds that generalize results for doubly robust learning in the single-step setting, and show promising empirical performance in several different contexts. Our approach is practical for policy optimization, and does not need any structural (e.g., Markovian) assumptions.
\end{abstract}

\section{Introduction}

The promise of personalized data-driven decision-making has led to a surge in interest in
methods that leverage observational data to help inform 
how and to whom interventions should be applied
\citep*{athey2017efficient,
	bertsimas2020predictive,
	dudik2014doubly,
	elmachtoub2017smart,
	kallus2018confounding,
	kitagawa2015should,
	manski2004statistical,
	swaminathan2015batch,
	zhang2012estimating,
	zhao2012estimating}. 
Any solution to this policy learning problem needs to deal with numerous
difficulties, including how to incorporate robustness to potential selection bias
as well as fairness constraints articulated by stakeholders, and there have been
several notable advances that address these difficulties over the past few years.

One limitation of this line of work, however, is that the results cited above all focus on
a static setting where a decision-maker only sees each subject once and immediately
decides how to treat the subject. In contrast, many problems of applied interest
involve a dynamic component whereby the decision-maker makes a series of 
decisions based on time-varying covariates. In medicine, if a patient has a disease for which
all known cures are invasive and have serious side effects, their doctor may choose to
monitor disease progression for some time before prescribing one of these invasive
treatments. As another example, a health inspector needs to not only choose which restaurants
to inspect, but also when to carry out these inspections. 

In this paper, we study the problem of learning dynamic when-to-treat policies,
where a decision-maker is only allowed to act once, but gets to choose both which
action to take and when to perform the action.\footnote{We note that the policies of interest in this paper also include when-to-stop policies. By flipping the treatment indicator, it is without loss of generality that we only consider when-to-treat policies.}
This setting covers several application
areas that have recently been discussed in the literature, including when to start
antiretroviral therapy for HIV-positive patients to prevent AIDS while mitigating side
effects \citep{when2009timing}, when to recommend mothers to stop breastfeeding
to maximize infants' health \citep*{moodie2009estimating}, and when to
to turn off ventilators for intensive care patients to maximize health outcomes \citep{prasad2017reinforcement}.

The available literature has developed several methods for evaluating and learning dynamic
treatment rules from prior data, with notable contributions from statistics and epidemiology communities including from  \cite{murphy2003optimal,robins2004optimal,murphy2005generalization, luckett2019estimating,tsiatis2019dynamic,zhang2013robust,zhang2018interpretable}, 
\citet[Chapter 4]{van2018targeted}, and 
the batch reinforcement learning community such as  \citet{jiang2015doubly} and \citet{thomas2016data}. 
As discussed further below, these papers develop general approaches that can be used with arbitrary
dynamic treatment rules. Here, in contrast, we seek to exploit special structure of the when-to-treat
problem to develop tailored learning methods with desirable statistical and computational properties. 

In developing our approach, we build on recent results on doubly robust
static policy learning \citep*{athey2017efficient, zhou2018offline},
and show how they can be adapted to our dynamic setting
without making any structural (e.g., Markovian) assumptions
and without compromising computational performance. 
Throughout this paper, we assume sequential ignorability, meaning that
any confounders that affect making a treatment choice at time $t$ have already been
measured by time $t$. Sequential ignorability is a widely used generalization of
the classical ignorability assumption of \citet{rosenbaum1983central} to the dynamic
setting \citep*{hernan2001marginal,murphy2003optimal,robins1986new,robins2004optimal}.
We develop methods that can leverage generic machine learning estimators of
various nuisance components (e.g., the propensity of starting treatment in any given
state and time) for learning policies with strong utilitarian regret bounds that hold
in a nonparametric setting.

Our problem setting is closely related to batch reinforcement learning \citep{sutton2018reinforcement}.
The types of guarantees we derive, however, are more closely related to results
from the static policy learning literature, in that we use tools from semiparametric statistics to derive sharp regret bounds given
only nonparametric assumptions. To our 
knowledge, the reinforcement learning literature has not  
pursued nor obtained the type of results we achieve here
for off-policy policy learning
in a nonparametric setting.

We also note work on optimal stopping motivated by the problem
of when to buy or sell an asset. This setting, however, is different from ours in that most
of the literature on optimal stopping either works with a known probabilistic model
\citep{jacka1991optimal,van1976optimal}, or assumes that we can observe the price
evolution of the asset whether or not we purchase it \citep*{goel2017sample}. In contrast,
we work in a nonparametric setting, and adopt a potential outcomes model
in which we only get to observe outcomes corresponding to the sequence of actions we
choose to take \citep{imbens2015causal,robins1986new}. \citet{rust1987optimal} considers the descriptive
problem of fitting an optimal stopping model to the behavior of a rational agent;
this is different from the problem of learning a decision rule that
can be used to guide future decisions in this paper. We will review the related literature in more detail
in Section \ref{sec:related-work} after first presenting our method below.

\section{Policy Learning under Sequential Ignorability}

\subsection{Setup and Notation}
\label{sec:notation}

We work in the following statistical setting. We observe a set of $i = 1, \, \ldots, \, n$ independent
and identically distributed trajectories generated from some distribution $\mathbb{P}$ that describe the evolution of subjects over $T$ time steps.
For each subject $i$, we observe a vector of states \smash{$S^{(i)} \in \set^T$} and actions \smash{$A^{(i)} \in \acal^T$},
as well as a final outcome \smash{$Y^{(i)} \in \RR$}.\footnote{We note that it is without loss of generality that we assume the outcome $Y^{(i)}$ is only observed at the end of a trajectory, since intermediate outcomes/rewards can be incorporated as part of the state representation.} For each $t = 1, \, \ldots, \, T$, \smash{$S_t^{(i)}$} denotes the state of
the subject at time $t$ and \smash{$A_t^{(i)}$} denotes the action taken. We write the set of possible actions
as $\acal = \{0, 1, \cdots, K\}$, and let $A_t=0$ denote no action (i.e., no treatment assignment) at time $t$.
We define the filtration $\ff_1 \subseteq \ff_2 \subseteq \cdots \subseteq \ff_{T+1}$,
where $\ff_{t} = \sigma\p{S_{1:t}, \, A_{1:t-1}}$ contains information available at time $t$
for $t = 1, \, \ldots, \, T$, and $\ff_{T+1}  = \sigma\p{S_{1:T}, \, A_{1:T},\, Y}$ also has
information on the final outcome.
For notational convenience, we denote \smash{$S_{t_1:t_2}^{(i)} := \{S_{t_1}^{(i)}, \cdots, S_{t_2}^{(i)}\}$},
and we similarly define \smash{$A_{t_1:t_2}^{(i)}$}, and
write the relevant generalization of the propensity score as
$e_{t,a}(s_{1:t}) = \PP{A_t = a \cond S_{1:t} = s_{1:t}, \, A_{1:(t-1)} = 0}$.

We formulate causal effects in terms of potential outcomes
\citep{neyman1923applications, robins1986new, rubin1974estimating}. For any set
of actions $a \in \acal^T$, we posit potential outcomes \smash{$Y^{(i)}(a_{1:T})$} and \smash{$S_t^{(i)}(a_{1:(t-1)})$} corresponding to the
outcome and state values we would have obtained for subject $i$ had we assigned treatment sequence $a$.
In order to identify causal effects, we make the standard assumptions of
sequential ignorability, consistency and overlap \citep*{hernan2001marginal,murphy2003optimal,robins1986new,robins2004optimal}.

\begin{assu}[Consistency of potential outcomes] \label{assu:consistency}
	Our observations are consistent with potential outcomes, in the sense that
	\smash{$Y^{(i)} = Y^{(i)}(A_{1:T}^{(i)})$} and \smash{$S_t^{(i)} = S_t^{(i)}(A_{1:(t-1)}^{(i)})$}. 
\end{assu}

\begin{assu}[Sequential Ignorability]\label{assu:ignorability}
	Actions cannot respond to future information, i.e.,
	\smash{$\{Y(A_{1:(t-1)}, \, a_{t:T}), \, S_{t'}(A_{1:(t-1)}, \, a_{t:(t'-1)})\}_{t'=t+1}^T \indep A_t \cond \ff_t$} for all $t= 1,\cdots,T$.
\end{assu}

\begin{assu}[Overlap]\label{assu:overlap}
	There are constants $0 < \eta, \, \eta_0 < 1$ such that, for all $t=1,\cdots, T$
	and $s_{1:t} \in \mathcal{S}^t$, the following hold:
	$e_{t,a}(s_{1:t}) > \eta/T$ for all $a \in \mathcal{A} \setminus \cb{0}$
	and $ e_{t,0}(s_{1:t}) > 1 - \eta_0/T$.
\end{assu}

A policy $\pi$ is a function that, for each time $t = 1, \, \ldots, \, T$, maps time-$t$ observables to an action:
\smash{$\pi_t: \scal^t \times \acal^{t-1} \to \acal$} such that $\pi_t$ is $\ff_t$-measurable;
then $\pi := \{\pi_t\}_{t=1}^T$.
Recall that we focus on when-to-treat type rules, meaning that the decision-maker only
gets to act once by starting a non-0 treatment regime at the time of their choice. For
example, if $K = 3$ and $T = 5$, then the decision-maker may choose for instance to start treatment option \#2
at time $t = 4$, resulting in a trajectory $A = (0, \, 0, \, 0, \, 2, \, 2)$.

When applying $\pi$ on-policy, the behavior of $\pi$ is fully characterized by
the time at which $\pi$ chooses to act, denoted by $\tau_{\pi} = \inf\cb{t : \pi_t(\cdot, \cdot) \neq 0}$,  and the treatment chosen, denoted by 
$W_{\pi} = \pi_{\tau_{\pi}}(\cdot, \cdot)$. When $\pi$ chooses to never initiate treatment, we use the convention that $\tau_{\pi} = T+1$ and $W_{\pi} = 0$. Note that $\tau_\pi$ and $W_\pi$ are both $\ff_{\tau_\pi}$-measurable.\footnote{See also \citet{athey2018design} for a closely related
	discussion of potential outcomes in the context of staggered adoption.}  For completeness, we also need to specify
how $\pi$ behaves off-policy, i.e., how $\pi$ would prescribe treatment along a trajectory
whose past action or treatment sequence may disagree with $\pi$; and, in this paper, we do so by assuming that $\pi$
is regular in the sense of Definition \ref{def:regular_policy} below.

\begin{defi}[Regular policy]
	\label{def:regular_policy}
	A regular when-to-treat policy $\pi$ is determined by an $\ff_t$-measurable stopping time $\tau_\pi$
	and an associated $\ff_{\tau_\pi}$-measurable decision variable $W_\pi \in \cb{1, \, \ldots, \, K}$ as follows:
	For each time $t = 1, \, \ldots,\, T$,
	if $A_{t - 1} \neq 0$ then $\pi_t(S_{1:t}, \, A_{1:(t-1)}) = A_{t-1}$,
	else if $t \geq \tau_\pi$ then $\pi_t(S_{1:t}, \, A_{1:(t-1)}) = W_\pi$,
	else $\pi_t(S_{1:t}, \, A_{1:(t-1)}) = 0$.
\end{defi}

The main substance of Definition \ref{def:regular_policy} is that we assume that if $\pi$ suggests to start treatment $k$ at a given moment, it persists in this choice $k$ even if we fail
to start treatment immediately. One notable limitation of this regularity condition is in the setting
where some patients may die or otherwise be unable to receive treatment.\footnote{For example,
	consider a case where $\pi$ says we should have started treatment on day 5 but we didn't in fact
	start treatment then, i.e., $A_5 = 0$, and then the patient dies on day 6. Here, realistically,
	$\pi$ should recognize that starting treatment is now impossible and prescribe
	$\pi_6(S_{1:6}, A_{1:5}) = 0$; however, doing so would be inconsistent with Definition \ref{def:regular_policy}.}
We discuss extensions of our approach beyond regular policies in Section \ref{sec:death}.

Following \citet{murphy2005generalization}, we let $f_t(S_t \cond S_{1:(t-1)}, A_{1:(t-1)})$ be the conditional density for state transitions at time $t$. We can write the distribution function for trajectories $(s_{1:T}, a_{1:T})$ as
\begin{align}
\label{eq:product_form}
f(s, \, a) = f(s_1) \, \PP{A_1 \cond s_1} \prod_{t=2}^T f_t\p{s_t \cond s_{1:(t-1)}, a_{1:(t-1)}} \PP{A_t=a_t \cond s_{1:t}, a_{1:(t-1)}},
\end{align}
and denote the expectation with respect to this distribution as $\mathbb{E}$. Similarly, the distribution of a trajectory under policy $\pi$ is 
\begin{align}
\label{eq:product_form_pi}
f(s, \, a; \, \pi) = f(s_1)  \mathbbm{1}_{a_1 = \pi_1(s_1)} \prod_{t=2}^T f_t\p{s_t \cond s_{1:(t-1)}, a_{1:(t-1)}} \mathbbm{1}_{a_t=\pi_t(s_{1:t}, a_{1:(t-1)})},
\end{align}
and we use $\mathbb{E}_\pi$ to denote the expectation with respect to it. Define 
\begin{equation}
V_\pi := \EE[\pi]{Y}= \EE{Y\p{\pi_1(S_1), \pi_2\p{S_1, \, S_2(\pi_1(S_1)), \pi_1(S_1)}, \ldots}}
\end{equation}
to be the value of the policy $\pi$, i.e., the expected outcome $Y$ with actions $A_t$ chosen
according to $\pi$ such that $A_t = \pi_t(S_{1:t}, \, A_{1:(t-1)})$ for all $t = 1, \, \ldots, \, T$.
We further define the conditional value function
\begin{align}
&\mu_{\pi,t}(s_{1:t}, a_{1:t-1}) \nonumber\\
&= \EE[\pi]{Y \cond S_{1:t} = s_{1:t}, A_{1:t-1} = a_{1:t-1}} \nonumber \\
&= \EE{Y(a_{1:t-1}, \pi_{t}(S_{1:t}, A_{1:t-1}),  \pi_{t+1}(\ldots),\,\ldots) \cond S_{1:t}=s_{1:t}, A_{1:t-1}=a_{1:t-1}},
\end{align}
and the $Q$-function
\begin{align}
&Q_{\pi,t}(s_{1:t}, a_{1:t}) \nonumber\\
&= \EE[\pi]{Y \cond S_{1:t}=s_{1:t}, A_{1:t} = a_{1:t}}\nonumber \\
&= \EE{Y(a_{1:t}, \pi_{t+1}(S_{1:t}, S_{t+1}(A_{1:t}),A_{1:t}),\,\ldots) \cond S_{1:t}=s_{1:t}, A_{1:t}=a_{1:t}}.
\end{align}
For any class $\Pi$, we write the
optimal value function as $V^* = \sup_{\pi \in \Pi} V_\pi$,
and define the regret of a policy $\pi \in \Pi$ as $R(\pi) = V^* - V_{\pi}$ \citep{manski2004statistical}.
Given this setting, our goal is to learn the best policy from a predefined policy class $\Pi$ to minimize regret.
Our main result is a method for learning a policy $\hpi \in \Pi$ along with
a bound on its regret $R(\hpi)$.

\subsection{Existing Methods}
\label{sec:existingwork}
In the static setting, a popular approach to policy learning starts by first providing
an estimator $\hV_\pi$ for the value $V_\pi$ of each feasible policy $\pi \in \Pi$, and
then sets \smash{$\hpi = \argmax\{\hV_\pi : \pi \in \Pi\}$}
\citep*[e.g.,][]{athey2017efficient,
	kitagawa2015should,
	manski2004statistical,
	swaminathan2015batch,
	zhang2012estimating}.
At a high level our goal is to pursue the same strategy, but now in a dynamic setting.
The challenge is then to find a robust estimator $\hV_\pi$ that behaves well
when optimized over a policy class $\Pi$ of interest---both statistically and computationally. 

Perhaps the most straightforward approach to estimating $V_\pi$  starts from inverse propensity weighting
as used in the context of marginal structural modeling \citep*{robins2000marginal, precup2000eligibility}.
Given sequential ignorability, we can write inverse propensity weights \smash{$\gamma_t^{(i)}(\pi)$}
for any policy $\pi$ recursively as follows, resulting in a value estimate
\begin{align}
\label{eq:IPW}
\hat{V}_\pi^{\textrm{IPW}} = \frac{1}{n} \sum_{i=1}^n \gamma_T^{(i)}(\pi) Y^{(i)}, \ \
\gamma_t^{(i)}(\pi) =  \frac{\gamma_{t-1}^{(i)} \mathbbm{1}\p{\cb{A_t^{(i)} = \pi(S_{1:t}^{(i)}, A_{1:t-1}^{(i)})}}}{\PP{A_t^{(i)} = \pi(S_{1:t}^{(i)}, A_{1:t-1}^{(i)}) \cond S_{1:t}^{(i)}, A_{1:t-1}^{(i)}}},
\end{align}
or the normalized alternative
$ \hat{V}_\pi^{\textrm{WIPW}} = \sum_{i=1}^n \gamma_T^{(i)}(\pi) Y^{(i)} \,\big/\, \sum_{i=1}^n \gamma_T^{(i)}(\pi)$.
The functional form of \smash{$\hat{V}_\pi^{\textrm{IPW}}$} makes it feasible to optimize this value estimate over
a pre-specified policy class $\pi \in \Pi$ (e.g., via a grid-search or mixed integer programming). By Assumptions
\ref{assu:consistency}, \ref{assu:ignorability} and \ref{assu:overlap}, IPW is consistent if the treatment probabilities are known a-priori, and by uniform concentration arguments
following \citet{kitagawa2015should}, the regret of the policy $\hpi$ learned by maximizing
\smash{$\hat{V}_\pi^{\textrm{IPW}}$} over $\pi \in \Pi$ decays as $1/\sqrt{n}$ if $\Pi$ is not too large (e.g., if $\Pi$ is a VC-class).

While inverse propensity weighting is a simple and transparent approach to estimating $V_\pi$,
it has several limitations. In observational studies treatment probabilities need to be estimated
from data, and it is known that the variant of \eqref{eq:IPW} with estimated weights \smash{$\hgamma_t^{(i)}(\pi)$}
can perform poorly with even mild estimation error \citep[see, e.g.,][]{liu2018representation}. Furthermore, for any policy $\pi$ considered, the IPW
value estimator only uses trajectories that match the policy $\pi$ exactly, which can make policy
learning sample-inefficient. Finally, the IPW estimator is known to be unstable when treatment propensities are small,
and this difficulty is exacerbated in the multi-period setting as the probability of observing any specific trajectory decays.
In the static (i.e., single time step) policy learning setting, related considerations led several
authors to recommend against inverse propensity weighted policy learning and to develop new
methods that were found to have stronger properties both in theory and in practice
\citep*{athey2017efficient, dudik2014doubly, kallus2018balanced, zhang2012estimating, zhou2015residual}.

Another approach to estimating $V_\pi$ is using a doubly robust (DR) estimator
as follows \citep*{jiang2015doubly,thomas2016data,zhang2013robust}
\begin{equation}
\label{eq:DR}
\hV_\pi^{\textrm{DR}} = \frac{1}{n} \sum_{i = 1}^n \p{ \hgamma_T^{(i)}(\pi) Y^{(i)}  
	- \sum_{t = 1}^T \p{\hgamma_t^{(i)}(\pi) - \hgamma_{t-1}^{(i)}(\pi)} \hmu_\pi\p{S_{1:t}^{(i)}, A_{1:t-1}^{(i)}}},
\end{equation} 
where \smash{$\hmu_\pi(\cdot)$} is an estimator of \smash{$\mu_\pi(\cdot)$},
the expected value following policy $\pi$ conditionally on the history up to time $t$
as defined in the previous subsection. This estimator generalizes the well known
augmented inverse propensity weighted estimator of  \citet*{robins1994estimation}
beyond the static case.
The doubly robust estimator \eqref{eq:DR} is consistent if either the propensity weights \smash{$\{\hgamma_t(\cdot)\}_{t=1}^T$} or
the conditional value estimates \smash{$\hmu_\pi(\cdot)$} are consistent.
For a further discussion of doubly robust estimation under sequential ignorability,
see \citet[Chapter 4]{van2018targeted}, \citet*{tsiatis2019dynamic}, and references therein.

From an optimization point of view, a limitation of \eqref{eq:DR} is
that evaluating a given policy $\pi$ requires nuisance components estimates
\smash{$\hmu_\pi(\cdot)$} that are specific to the policy under consideration.
This makes policy learning by optimizing \smash{$\hV_\pi^{\textrm{DR}}$}
problematic for several reasons. Computationally, maximizing \smash{$\hV_\pi^{\textrm{DR}}$}
for all $\pi$ in a non-trivial set $\Pi$ would require solving a multitude of non-parametric dynamic
programming problems. 

Perhaps even more importantly, statistically, standard regret bounds for policy learning for single time step problems rely crucially on the fact that \smash{$\hV_\pi$} is continuous in $\pi$ in an appropriate sense, meaning that
two policies are taken to have similar values if they make similar recommendations
in almost all cases \citep[see, e.g.,][]{athey2017efficient}. But, if \smash{$\hmu_\pi(\cdot)$} is learned separately
for each $\pi$, we have no reason to believe that two similar policies would necessarily have
similar value function estimates---unless one were to use specially designed \smash{$\hmu_\pi(\cdot)$}
estimators.\footnote{One heuristic solution to this difficulty, proposed by \citet{zhang2013robust},
	is to first derive a policy estimate $\hpi^*$ via, e.g., IPW or fitted-$Q$ learning, and then to use value
	estimates \smash{$\hmu_{\hpi^*}(\cdot)$} to evaluate all policies $\pi \in \Pi$ using \eqref{eq:DR}.
	The advantage of this proposal is that learning by maximizing \smash{$\hV_\pi^{\textrm{DR}}$} becomes
	more tractable, since one does not need to re-fit nuisance components in order to evaluate different policies.}

\subsection{Advantage Doubly Robust Policy Learning}
\label{sec:proposal}

The goal of this paper is to develop a new method for policy learning that
addresses the shortcomings of both inverse propensity weighting and the doubly
robust method discussed above in the case of when-to-treat policies.\footnote{We emphasize that the IPW and DR estimators discussed
	above can be used with general dynamic policies; in contrast, our method can only be used for learning
	when-to-treat policies. Our proposed method does not present an alternative to IPW or DR in the general case.}
Our main proposal, the Advantage Doubly Robust (ADR)
estimator, uses an outcome regression like the doubly robust estimator \eqref{eq:DR}
to stabilize and robustify its value estimates. However, unlike the estimator \eqref{eq:DR}
which needs to use different outcome regressions \smash{$\hmu_\pi(\cdot)$} to evaluate
each different policy $\pi$, ADR only has ``universal'' nuisance components that do not
depend on the policy being estimated, leveraging the when-to-treat (or when-to-stop)
structure of the domain. Throughout this paper, we will find that this
universality property enables us to both effectively optimize our value estimates to learn
policies and to prove robust utilitarian regret bounds.

The motivation for our approach starts from an ``advantage decomposition'' presented below.
First, define
\begin{align}
\label{eq:mu_now_next}
\begin{split}
&\mu_{now, k}(s_{1:t}, \, t) := \EE{Y \cond S_{1:t}=s_{1:t}, A_{1:t-1}=0, A_t=k},  \\
&\mu_{next, k}(s_{1:t}, \, t) := \EE{ \mu_{now, k}\p{S_{1:t+1}, \, t+1} \! \cond S_{1:t}=s_{1:t}, A_{1:t}=0},
\end{split}
\end{align}
which measure the conditional value of a policy that starts treatment $k$ either now
or in the next time period, given that we have not yet started any treatment.
Note that, for any when-to-treat policy $\pi$ as considered in this paper, the expectations
in \eqref{eq:mu_now_next} do not depend on $\pi$ because the conditioning specifies all
actions from time $t = 1$ to $T$. Given policies $\pi, \, \pi' \in \Pi$, define $\Delta(\pi, \pi') = V_\pi - V_{\pi'}$ to be the difference in value of the two policies. Denote the never-treating policy by $\bm{0}$. 
Then, a result from  \citet[Chapter 5]{kakade2003sample} and \citet{murphy2005generalization}
yields the following.

\begin{lemm}\label{lemm:murphy-decomp}
	Under Assumptions \ref{assu:consistency} and \ref{assu:ignorability} 
	let $\pi$ be a regular when-to-treat policy in the sense of Definition \ref{def:regular_policy}.
	Then
	\begin{equation}
	\label{eq:murphy-decomp}
	\Delta(\pi, \bm{0}) = \EE[\bm{0}] {\sum_{t=\tau_\pi}^T \mu_{now, W_\pi}(S_{1:t}, t) - \mu_{next, W_\pi}(S_{1:t}, t)},
	\end{equation}
	where $\mathbb{E}_{\bm{0}}$ samples trajectories under a never-treating policy and,
	following Definition \ref{def:regular_policy}, $\tau_\pi$ is the time at which $\pi$ starts treating
	and $W_\pi$ is the treatment chosen at that time.
	\begin{proof}
		Given our setup, Lemma 1 of \citet{murphy2005generalization} implies that 
		\begin{equation}
		\label{eq:murphy}
		\begin{split}
		\Delta(\pi, \bm{0}) &= -\EE[\bm{0}] {\sum_{t=1}^T Q_{\pi,t}(S_{1:t}, A_{1:t}) - \mu_{\pi,t}(S_{1:t}, A_{1:t-1})} \\
		&= -\EE[\bm{0}] {\sum_{t=1}^T \mathbbm{1}_{t \geq \tau_\pi} \p{Q_{\pi,t}(S_{1:t}, \bm{0}_{1:t}) - \mu_{\pi,t}(S_{1:t}, \bm{0}_{1:(t-1)})}}.
		\end{split}
		\end{equation}
		Because $\pi$ is a regular when-to-stop policy, whenever $t \geq \tau_\pi$,
		the policy $\pi$ prescribes starting treatment $W_\pi$ immediately if no other treatment has been started yet, i.e.,
		\begin{align*}
		\mathbbm{1}_{t \geq \tau_\pi} \, \mu_{\pi,t}(S_{1:t}, \bm{0}_{1:(t-1)}) 
		&\stackrel{(a)}{=}  \mathbbm{1}_{t \geq \tau_\pi} \EE{Y( \bm{0}_{1:(t-1)}, W_\pi, W_\pi, \ldots) \cond S_{1:t}, A_{1:t-1}= \bm{0}_{1:(t-1)}}\\
		&\stackrel{(b)}{=}  \mathbbm{1}_{t \geq \tau_\pi} \EE{Y( \bm{0}_{1:(t-1)}, W_\pi, W_\pi, \ldots) \cond S_{1:t}, A_{1:t-1}= \bm{0}_{1:(t-1)}, A_t=W_\pi}\\
		&\stackrel{(c)}{=}   \mathbbm{1}_{t \geq \tau_\pi} \EE{Y \cond S_{1:t}, A_{1:t-1}= \bm{0}_{1:(t-1)}, A_t=W_\pi}\\
		&= \mathbbm{1}_{t \geq \tau_\pi}\, \mu_{now, W_\pi}(S_{1:t}, t),
		\end{align*}
		where (a), (b) and (c) are by Definition \ref{def:regular_policy}, Assumption \ref{assu:ignorability}, and Assumption \ref{assu:consistency} respectively.
		Furthermore, given our definition of regular policies, we know that if $t \geq \tau_\pi$ and $A_t = 0$, then $\pi$ will deterministically
		prescribe treatment $W_{\pi}$ and time $t+1$ regardless of the state $S_{t+1}$, and so
		\begin{equation}
		\label{eq:murphy_devel}
		\begin{split}
		&\mathbbm{1}_{t \geq \tau_\pi}\, Q_{\pi,t}(S_{1:t}, \, \bm{0}_{1:t}) \\
		&\ \ \ \ \ \ \ = \mathbbm{1}_{t \geq \tau_\pi}\, \int \EE[\pi]{Y \cond S_{1:t+1}, \, A_{1:t} = \bm{0}_{1:t}} dF_{t+1}\p{S_{t+1} \cond S_{1:t}, \, A_{1:t} = \bm{0}_{1:t}} \\
		&\ \ \ \ \ \ \ \stackrel{(d)}{=} \mathbbm{1}_{t \geq \tau_\pi}\, \int \EE{Y(\bm{0}_{1:t}, W_\pi, W_\pi, \ldots) \cond S_{1:t+1}, \, A_{1:t} = \bm{0}_{1:t}} dF_{t+1}\p{S_{t+1} \cond S_{1:t}, \, A_{1:t} = \bm{0}_{1:t}} \\
		&\ \ \ \ \ \ \ \stackrel{(e)}{=} \mathbbm{1}_{t \geq \tau_\pi}\, \int \EE{Y(\bm{0}_{1:t}, W_\pi, W_\pi, \ldots) \cond S_{1:t+1}, \, A_{1:t} = \bm{0}_{1:t}, A_{t+1}=W_\pi} \\
		&\ \ \ \ \ \ \ \ \ \ \ \ \ \times dF_{t+1}\p{S_{t+1} \cond S_{1:t}, \, A_{1:t} = \bm{0}_{1:t}} \\
		&\ \ \ \ \ \ \  \stackrel{(f)}{=}  \mathbbm{1}_{t \geq \tau_\pi}\, \int \EE{Y\cond S_{1:t+1}, \, A_{1:t} = \bm{0}_{1:t}, A_{t+1}=W_\pi} dF_{t+1}\p{S_{t+1} \cond S_{1:t}, \, A_{1:t} = \bm{0}_{1:t}},
		\end{split}
		\end{equation}
		where (d), (e) and (f) are by Definition \ref{def:regular_policy}, Assumption \ref{assu:ignorability} and \ref{assu:consistency} respectively. The conclusion \eqref{eq:murphy-decomp} emerges by plugging these facts into \eqref{eq:murphy}.
	\end{proof}
\end{lemm}

In Lemma \ref{lemm:murphy-decomp} the expectation is taken with respect to the never-treating policy $\bm{0}$. 
To make this result usable in practice, the following lemma translates it in terms of expectations taken with
respect to the sampling measure.
Recall that the propensity of starting treatment $a$ assuming a never-treating history up to time $t$ is denoted by $e_{t,a}(s_{1:t}) = \PP{A_t = a \cond S_{1:t}=s_{1:t}, A_{1:t-1} = 0}$.
The proof of Lemma \ref{lemm:is-decomp}, given in the appendix, follows directly from a change of measure.

\begin{lemm} \label{lemm:is-decomp} In the setting of Lemma \ref{lemm:murphy-decomp},
	\begin{align}\label{eq:is-adv-decomp}
	\Delta(\pi, \bm{0}) = \EE {\sum_{t=1}^T \mathbbm{1}_{t \geq \tau_\pi}\frac{\mathbbm{1}_{A_{1:t-1}=0}}{\prod_{t'=1}^{t-1}e_{t',0}(S_{1:{t'}})} \p{\mu_{now, W_\pi}(S_{1:t}, t) - \mu_{next, W_\pi}(S_{1:t}, t)}}.
	\end{align}
\end{lemm}

This representation \eqref{eq:is-adv-decomp} is at the core of our approach, as it decomposes the
relative value of any given policy $\pi$ in comparison to that of the never-treating policy $\bm{0}$ into a
sum of local advantages. For any $t$, the local advantage 
\begin{equation}
\label{eq:delta}
\delta_{local, k}(s_{1:t}, t) := \mu_{now,k}(s_{1:t}, t) - \mu_{next,k}(s_{1:t}, t)
\end{equation}
is the relative advantage of starting treatment $k$ at $t$ versus at $t+1$ given the the state history $s_{1:t}$.
The upshot is that the specification of these local advantages does not depend on which policy we are evaluating,
so if we get a handle on quantities $\delta_{local, k}(s_{1:t}, t)$ for all $s$ and $t$, we can use \eqref{eq:is-adv-decomp}
to evaluate any policy $\pi$.

Note that the quantity defined in \eqref{eq:delta} can be seen as a specific treatment effect, namely the effect of starting
treatment $k$ at time $t$ versus $t+1$ among all trajectories that were in state $s_{1:t}$ at time $t$ and started treatment
$k$ in either time $t$ or $t+1$. Given this observation, we propose turning \eqref{eq:is-adv-decomp} into a feasible
estimator by replacing all instances of the unknown regression surfaces $\delta_{local, k}(s_{1:t}, t)$ with doubly robust
scores analogous to those used for augmented inverse propensity weighting in the static case \citep{robins1995semiparametric}.

More specifically, we propose the following 3-step policy learning algorithm, outlined as Algorithm \ref{alg:adr}.
We call our approach the Advantage Doubly Robust (ADR) estimator, because it replaces local advantages \eqref{eq:delta}
with appropriate doubly robust scores \eqref{eq:dr_score} when estimating $\Delta(\pi, \bm{0})$.
In the first estimation step in Algorithm \ref{alg:adr}, we employ cross-fitting where we divide the data into $Q$ folds, and only use the $Q-1$ folds
that a sample trajectory does not belong to to learn the estimates of its nuisance components; we use superscript $-q(i)$ on a
predictor to denote using trajectories of all folds excluding the fold that the $i$-th trajectory belongs to in training a
predictor.\footnote{The idea of cross-fitting has gained growing popularity recently to reduce the
	effect of own-observation bias and to enable results on semiparametric rates of convergence using generic
	nuisance component estimates \citep{athey2017efficient, chernozhukov2016double, schick1986asymptotically}.}

\RestyleAlgo{boxruled}
\LinesNumbered
\begin{algorithm}[t]
	\caption{\textbf{Adantage Doubly Robust (ADR) Estimator}\label{alg:adr}}
	Estimate the outcome models $\mu_{now,k}(\cdot)$,
	$\mu_{next,k}(\cdot)$, as well as treatment propensities $e_{t,a}(s_{1:t})$ with cross fitting using any supervised learning method tuned for prediction accuracy.
	
	Given these nuisance component estimates, we construct value estimates 
	\begin{equation}
	\label{eq:feasible-adr}
	\hat{\Delta}(\pi, \bm{0})
	= \frac{1}{n} \sum_{i=1}^{n} \sum_{t=1}^T \mathbbm{1}_{t \geq \tau_\pi^{(i)}} \frac{\mathbbm{1}_{A_{1:t-1}^{(i)} = 0}}{\prod_{t'=1}^{t-1}\hat{e}^{-q(i)}_{t',0}(S_{1:{t'}}^{(i)})}\hat{\Psi}_{t,W_\pi}(S_{1:t}^{(i)})
	\end{equation}
	for each policy $\pi \in \Pi$,  	where the relevant doubly robust score is
	\begin{equation}
	\label{eq:dr_score}
	\begin{split}
	\hat{\Psi}_{t,k}(S_{1:t}^{(i)}) &= \hat{\mu}^{-q(i)}_{now,k}(S_{1:t}^{(i)},t) - \hat{\mu}^{-q(i)}_{next,k}(S_{1:t}^{(i)}, t) \\
	&\ \ \ \ \ + \mathbbm{1}_{A_t^{(i)}=k} \frac{Y^{(i)} - \hat{\mu}^{-q(i)}_{now, k}(S_{1:t}^{(i)},t)}{\hat{e}^{-q(i)}_{t,k}(S_{1:t}^{(i)})} \\
	&\ \ \ \ \ - \mathbbm{1}_{A_t^{(i)}=0} \mathbbm{1}_{A_{t+1}^{(i)}=k} \frac{Y^{(i)} - \hat{\mu}^{-q(i)}_{next, k}(S_{1:t}^{(i)},t)}{\hat{e}^{-q(i)}_{t,0}(S_{1:t}^{(i)})\hat{e}^{-q(i)}_{t+1,k}(S_{1:t+1}^{(i)})}.
	\end{split}
	\end{equation}
	
	Learn the optimal policy by setting $\hat{\pi} = \argmax_{\pi \in \Pi} \hat{\Delta}(\pi, \bm{0})$.
\end{algorithm}

The main strength of this procedure relative to existing doubly robust approaches discussed above
\citep*{jiang2015doubly,thomas2016data,zhang2013robust} is that ADR can evaluate any stopping policy
using universal scores $\hat{\Psi}_{t,k}(\cdot)$ that do not depend on $\pi$. This allows us to ensure smoothness criteria we use to provide regret bounds for policy optimization. It also provides computational benefits for using the ADR estimator for policy optimization: 
the specific policy $\pi$ we are evaluating only enters
into \eqref{eq:feasible-adr} by specifying which doubly robust scores we should sum over.
In particular, the number of nuisance components we need to learn in the first step of the ADR procedure scales linearly with the horizon $T$, but not with the complexity of the policy class $\Pi$.

By constructing doubly robust scores $\hat{\Psi}_{t,k}(\cdot)$, the ADR estimator benefits from certain robustness properties;
however, it is not doubly robust in the usual sense, 
e.g., we do not robustly correct for the change of measure used to get from the representation in Lemma \ref{lemm:murphy-decomp}
to the one in Lemma \ref{lemm:is-decomp}.
We discuss the asymptotic behavior of our method in Section \ref{sec:asymptotics}.

In our experiments, we learn all the nuisance components in the first step with nonparametric regression methods (e.g., boosting, lasso, a deep net, etc.), and then optimize for the best in-class policy by performing a grid search over the parameters that define the policies in a policy class of interest.

\begin{rema}
	\label{rema:rewritenext}
	For the purpose of estimating $\mu_{next,k}(s_{1:t},t)$, it is helpful to re-express it in terms of conditional expectations.
	First, under Assumption \ref{assu:overlap}, we can continue from \eqref{eq:murphy_devel} and
	rewrite $\mu_{next, k}(s_{1:t},t)$ via inverse-propensity weighting as
	\begin{equation}
	\label{eq:rewritenext}
	\mu_{next, k}(s_{1:t}, t) = \EE{\frac{ \mathbbm{1}_{A_{t+1} = k} \, }{e_{t+1,k}(S_{1:t+1})}Y \cond S_{1:t}, \, A_{1:t} = \bm{0}_{1:t}}.
	\end{equation}
	Then, using Bayes' rule, we can verify that
	\begin{equation}
	\label{eq:weighted_reg}
	\mu_{next, k}(s_{1:t}, t) = \frac{\EE{Y /  e_{t+1, k}(S_{1:t+1})  \cond S_{1:t}=s_{1:t}, A_{1:t}=0, \, A_{t+1} = k}}
	{\EE{1 /  e_{t+1, k}(S_{1:t+1})  \cond S_{1:t}=s_{1:t}, A_{1:t}=0, \, A_{t+1} = k}}.
	\end{equation}
	This last expression implies that
	we can consistently estimate \smash{$\mu_{next,k}(\cdot,t)$} via weighted non-parametric regression of $Y$ on $S_{1:t}$
	on the set of observations with $A_{1:t}=0$ and $A_{t+1} = k$, with weights \smash{$e^{-1}_{t+1, k}(S_{1:t+1})$}.
	In practice, this may yield more stable estimates of $\mu_{next,k}(s_{1:t},t)$ than an unweighted non-parametric
	regression with response \smash{$\mathbbm{1}_{A_{t+1} = k} /  e_{t+1, k}(S_{1:t+1}) \, Y$}.
\end{rema}

\section{Related Work}
\label{sec:related-work}
The problem of learning optimal dynamic sequential decision rules is also called learning dynamic optimal regimes \citep{murphy2003optimal, robins2004optimal}, adaptive strategies \citep{lavori2000design}, or batch off-policy policy learning in the reinforcement learning (RL) literature \citep{sutton2018reinforcement}. There are a few predominant approaches: the G-estimation procedure \citep{robins1989analysis, robins1992g} learns the Structural Nested Mean Models (SNMM) \citep{robins1994correcting} which model the difference in the marginal outcome functions directly. See also \cite*{robins2004optimal, moodie2007demystifying,murphy2003optimal,  vansteelandt2014structural} for a further discussion.  \citet{orellana2006generalized} and \citet{van2007causal} proposed the dynamic marginal structural models (MSM) to model the marginal outcome function directly. Under MSM, \cite{robins1986new} proposed using the G-computation, which is a maximum likelihood approach for solving the MSM; \citet*{robins2000marginal} and \citet*{precup2000eligibility} proposed using the inverse propensity weighting (IPW) approach. Unlike G-estimation or G-computation, we focus on policy learning instead of structural parameter estimation, and our proposed approach ADR is more data efficient and robust compared to IPW. Finally, Q-learning\footnote{The fitted-Q iteration algorithm in reinforcement learning is a batch algorithm, and is also called Q-learning in the causal inference and biostatistics literature. This is not to be confused with Q-learning in the reinforcement learning literature, which is an online-version of the batch-mode fitted-Q algorithm.} \citep{watkins1992q} and the closely related fitted-Q iteration algorithm model the optimal marginal outcome $Q$ functions directly and seek to evaluate the optimal policy by stagewise backwards regression \citep{ernst2005tree,murphy2005generalization, prasad2017reinforcement,zhang2018interpretable}. 

Our work focuses on finding the best in-class policy in a pre-defined policy class, whereas fitted-Q iteration focuses on finding the best policy by learning $Q$ functions associated with the optimal policy, which might not fall into the predefined policy class. The two approaches are complementary, and the ADR estimator shines when there are predefined structural constraints on the policy class (e.g., for ease of interpretability, budget constraints, etc.). We note that the fitted-Q iteration algorithm can be adapted to learn the value of an arbitrary policy by learning the $Q$ functions associated with this policy. 
However, this makes optimization over a policy class intractable, as we would need to estimate separate $Q$ functions for each policy in the class.
For more discussion on the comparison of the existing approaches, see
\citet{chakraborty2013statistical}, \citet{moodie2007demystifying}, \citet{robins2008estimation},
\citet{tsiatis2019dynamic}, \citet{van2018targeted} and \citet{vansteelandt2014structural}.

Considerable progress has been made in learning good models for the value functions and combining them with propensity models in doubly robust forms. In reinforcement learning, there has been extensive work focused on learning good models \citep{farajtabar2018more,hanna2017bootstrapping,  liu2018representation}. \citet{guo2017using} focuse reducing the mean-squared error in policy evaluation in long horizon settings. \citet{ ernst2005tree}  and \citet{ormoneit2002kernel} study approximating the Bellman operator using empirical estimates with kernel averagers, and \citet{haskell2016empirical} focuses on the case with discrete state spaces. Recently \citet{doroudi2017importance} has shown that learning high-quality and fair policy decisions is nontrivial from inverse propensity weighting based policy evaluation methods. Adaptions of actor-critic \citep{degris2012off} and Gaussian processes \citep{schulam2017reliable} have been proposed for the off-policy setting as well. Finally, there has been a line of work that builds doubly-robust estimators that combines the model based estimators with inverse propensity weighting based estimates to improve robustness \citep*{langford2011doubly, jiang2015doubly,liu2018augmented, thomas2016data,  zhang2013robust,  zhao2014doubly}. We note that the closest works to ours are  \cite{jiang2015doubly}, \cite{thomas2016data} and \cite*{zhang2013robust}, with the key difference that our proposal has universal scores and nuisance components for all policies in a policy class, and thus is practical for policy optimization, whereas it is unclear yet if this is possible for generic Markov Decision Process settings considered in these works.

Among prior work from the reinforcement learning community that directly tries to learn an optimal value function and policy, formal bounds on the optimality of the resulting policy tend to require that the true value function is realizable by the regressor function used to model the value function in order to obtain good rates and consistent estimators \citep*[e.g.,][]{chen2019information,le2019batch,  munos2008finite}. Such results also require a bound on the concentratability coefficient \citep{munos2003error}, which measures the ratio of the state action distribution of a policy to the state-action distribution under the behavior policy, for any behavior policy \citep*[e.g.,][]{chen2019information,le2019batch,  munos2008finite}---this can be viewed as a similar analogue to the overlap assumption as in Assumption \ref{assu:overlap}.  Recent work provides convergence guarantees for batch policy gradient \citep{liu2019off}, but to our knowledge, there are no regret bounds on batch direct policy search and optimization.

In the special case of offline policy learning in a single timestep where a policy only needs to decide whether to treat a subject but not \textit{when}, substantial progress has been made in how to derive an optimal regret and performing optimization in finding the optimal policy \citep*{athey2017efficient,    kallus2018confounding, kitagawa2015should, swaminathan2015batch,zhang2012estimating,zhao2012estimating, zhou2018offline, zhou2015residual}. In particular, \cite{kitagawa2015should} establishes a lower bound $\Omega(1/\sqrt{n})$ on learning the regret in the case of binary treatment. \cite{athey2017efficient} and \cite{zhou2018offline} show a matching upperbound assuming the nuisance components can be learned at a much slower rate in the settings where the treatment consists of binary actions and multiple actions respectively. Extending this line of result to sequential multi-step settings is nontrivial for several reasons: First, it is unclear how to optimize efficiently across all policies in a policy class, especially given the increasing complexity with long horizons. Second, the regret results used in \cite{athey2017efficient} and \cite{zhou2018offline} rely on a chaining argument in which the special form of the estimator ensures values of policies close to each other is close. It is not obvious whether existing doubly-robust estimators in the sequential settings \citep*[e.g.,][]{thomas2016data,zhang2013robust} have such a form.

We note that there is a vast literature in optimal stopping \citep{goel2017sample, jacka1991optimal, mordecki2002optimal, van1976optimal}. In optimal stopping, the treatment choices are binary, i.e., whether to stop or not, and the goal is to optimize for a policy for when to start or stop a treatment. In our setup, we assume multiple treatment actions are allowed. 
Many existing works in optimal stopping (e.g., in finance) focus on the setup where a generator is available for the system dynamics, or the full potential outcomes are available in the training data. In our setup, we assume neither, and the policies of interest only make treatment decisions given data observed thus far. 

The problem of learning optimal decision rules is also closely related to learning heterogeneous treatment effects \citep*{athey2016recursive, athey2019generalized, chen2007large, kunzel2019metalearners, nie2017quasi, wager2018estimation}. In both problems, the goal is to learn individualized treatment effect and decision rules, but the type of estimands differ in that instead of learning a nonparametric function of the treatment effects, here we learn decision rules in a policy class.

Finally, we note that while we focus on the finite horizon setting, there is a large literature in policy evaluation in the infinite horizon setting (see \citet{antos2008fitted,antos2008learning,liu2018breaking,luckett2019estimating, munos2008finite} and references therein), and in the online setting (see \citet{shah2018q} and references therein).
We also note recent work of \citet{kallus2019double} on doubly robust methods for Markov decision processes.
These cases are considerably different from ours and are beyond the scope of this work. 

In this paper, we focus on the problem of making treatment decisions once and for all and with multiple actions at the decision point, which we note is a strict generalization of the setting in \cite{zhou2018offline}. We propose an advantage doubly robust estimator that draws upon the semiparametrics and orthogonal moments literature \citep{belloni2014inference,chernozhukov2016double,newey1994asymptotic,robins1995semiparametric, scharfstein1999adjusting}. There is also a growing number of existing works that have applied orthogonal moments construction to policy evaluation \citep{belloni2013program,chernozhukov2016locally, kallus2018balanced}. 

Finally, our proposed estimator heavily relies on an advantage decomposition in \cite{murphy2005generalization}. \cite{murphy2005generalization} focuses on the generalization error on a variant of Q-learning, and we turn such a decomposition into a practical and efficient estimator for learning policy values.

\section{Asymptotics}
\label{sec:asymptotics}

In this section, we study large-sample behavior of the advantage doubly robust estimator
proposed in Section \ref{sec:proposal} for policy learning in when-to-treat settings over a class of policies $\Pi$. It is now standard
in the literature in static policy learning for policies over a single decision to bound regret over the learned policy~\citep[e.g.,][]{athey2017efficient,kitagawa2015should,manski2004statistical,swaminathan2015batch}. However, to our knowledge there are no
directly comparable results for the sequential decision process setting.

Following the literature on static policy learning,  
our main goal is to prove a bound on the utilitarian regret $R$ of the learned policy $\hpi$, where
\begin{equation}
R\p{\hpi} = \sup\cb{V(\pi): \pi \in \Pi} - V\p{\hpi}.
\end{equation}
In order to do so, we follow the high-level proof strategy taken by \citet{athey2017efficient} for studying
static doubly robust policy learning. We first consider the behavior of an ``oracle'' learner who runs our procedure
but with perfect estimates of the nuisance components \smash{$\mu_{now,k}(\cdot)$},
\smash{$\mu_{next,k}(\cdot)$}, and \smash{$e_{t,k}(\cdot)$}, then we couple the behavior of our feasible estimator that uses estimated nuisance components with this oracle.

Following this outline, recall that our approach starts by estimating the policy value difference $\Delta(\pi, \bm{0})$ between
deploying policy $\pi$ and the never treating policy $\bm{0}$. The oracle variant of our estimator \smash{$\hDelta(\pi, \bm{0})$} is then
\begin{equation}
\label{eq:plugin-oracle-estimator}
\tilde{\Delta}(\pi, \bm{0})
= \frac{1}{n} \sum_{i=1}^{n} \sum_{t=1}^T \mathbbm{1}_{t \geq \tau_\pi} \frac{\mathbbm{1}_{A_{1:t-1}^{(i)} = 0}}{\prod_{t'=1}^{t-1}e_{t',0}(S_{1:{t'}})}\tilde{\Psi}_{t,W_\pi}(S_{1:t}^{(i)}),
\end{equation}
where
\begin{equation}
\label{eq:oracle-gamma}
\begin{split}
\tilde{\Psi}_{t,k}(S_{1:t}^{(i)}) &= \mu_{now,k}(S_{1:t}^{(i)},t) - \mu_{next,k}(S_{1:t}^{(i)}, t) \\
&\ \ \ \ \ + \mathbbm{1}_{A_t^{(i)}=k} \frac{Y^{(i)} - \mu_{now, k}(S_{1:t}^{(i)},t)}{e_{t,k}(S_{1:t}^{(i)})} \\
&\ \ \ \ \ - \mathbbm{1}_{A_{t}^{(i)}=0} \mathbbm{1}_{A_{t+1}^{(i)}=k} \frac{Y^{(i)} - \mu_{next, k}(S_{1:t}^{(i)},t)}{e_{t,0}(S_{1:t}^{(i)}) e_{t+1,k}(S_{1:t+1}^{(i)})}.
\end{split}
\end{equation}
We name \eqref{eq:plugin-oracle-estimator} the \emph{oracle} estimator since we assume
$\mu_{now,k}$, $\mu_{next,k}$, $e_{t,k}$ for $t=1,\cdots, T$ and $k = 1, \cdots, K$
take ground-truth values in \eqref{eq:oracle-gamma}.

Because the nuisance components in \eqref{eq:plugin-oracle-estimator} are known a-priori, we can
use a standard central limit theorem argument to verify the following.

\begin{lemm}
	\label{lemm:clt}
	Suppose that Assumptions \ref{assu:consistency} - \ref{assu:overlap} hold and that $\abs{Y} \leq M$ almost surely for some constant $M$, for a fixed policy $\pi \in \Pi$,
	\begin{equation}
	\label{eq:oracle_CLT}
	\begin{split}
	&\sqrt{n}(\tilde{\Delta}(\pi, \bm{0}) - \Delta(\pi,\bm{0})) \Rightarrow \nn(0, \Omega_\pi), \\ \where\ \ 
	&\Omega_\pi = 	\Var{\sum_{t=1}^T \mathbbm{1}_{t \geq \tau_\pi^{(i)}} \frac{\mathbbm{1}_{A_{1:t-1}^{(i)} = 0}}{\prod_{t'=1}^{t-1}e_{t',0}(S_{1:{t'}}^{(i)})}\tilde{\Psi}_{t,W_\pi}(S_{1:t}^{(i)})}.
	\end{split}
	\end{equation}
\end{lemm}

Next, we show that 
the rate of convergence suggested by \eqref{eq:oracle_CLT} is in fact uniform over the whole class $\Pi$ under appropriate bounded entropy conditions, thus enabling
a regret bound for the oracle learner that optimizes  \eqref{eq:plugin-oracle-estimator}.
To do so, we introduce some more notation. Let $H=\{S_{1:T}, A_{1:T}\}$ be the entire history of a trajectory.
Any policy $\pi$ that is regular in the sense of Definition \ref{def:regular_policy} can then be re-expressed as a mapping
from $H$ to a length $KT+1$ vector of all zeros except for an indicator $1$ at one position in the probability simplex,
such that\footnote{All regular policies can be expressed in the form \eqref{eq:pi-new-defn}; however, we emphasize
	that the the converse is not true: The form \eqref{eq:pi-new-defn} does not ensure that $\pi$ is $\ff_t$ measurable.}
\begin{equation}
\label{eq:pi-new-defn}
\pi(H) =  \begin{cases*}
v_{K(\tau_\pi - 1) + W_\pi} & if $\tau_\pi \leq T$, \\
v_{KT + 1} & else,
\end{cases*}
\end{equation}
where $v_m \in  \cb{0, \, 1}^{KT+1}$ is the indicator vector with the $m$-th position 1, and all others 0.
Given this form, we note that we can re-express \eqref{eq:oracle-gamma} as
\begin{equation}
\label{eq:delta-tilde-general}
\begin{split}
&\tilde \Delta(\pi, \pi')  = \frac{1}{n} \sum_{i=1}^n \langle \pi(H^{(i)}) -  \pi'(H^{(i)}), \, \tilde{\Gamma}^{(i)} \rangle, \ \text{where} \\
&\tilde{\Gamma}_{K(t-1)+k}^{(i)} = \sum_{t'=t}^T \frac{\mathbbm{1}_{A_{1:(t'-1)}^{(i)} = 0}}{\prod_{t''=1}^{t'-1} e_{t'',0}(S^{(i)}_{1:{t''}})}\tilde{\Psi}_{t,k}(S_{1:t'}^{(i)})
\end{split}
\end{equation}
for all $1\leq t\leq T$ and $1\leq k\leq K$, and $\tilde{\Gamma}_{KT+1}^{(i)} = 0$.

Given these preliminaries, let the Hamming distance between any two policies $\pi, \pi'$ be 
$$d_h (\pi, \pi') = \frac{1}{n} \sum_{i=1}^n \mathbbm{1}_{\pi\p{H^{(i)}} \neq \pi'\p{H^{(i)}}}.$$
Define the $\varepsilon$-Hamming covering number of $\Pi$ as $$N_{d_h}(\varepsilon, \Pi) = \sup \cb{N_{d_h}\p{\varepsilon, \Pi, \cb{H^{(1)}, \ldots,  H^{(n)}}} \cond H^{(1)}, \ldots, H^{(n)}},$$ where $N_{d_h}\p{\varepsilon, \Pi, \cb{H^{(1)}, \cdots,  H^{(n)}}}$ is the smallest number of policies $\pi^{(1)}, \pi^{(2)}, \ldots,  \in \Pi$ such that $\forall \pi \in \Pi, \exists \pi^{(i)}$ such that $d_h(\pi, \pi^{(i)}) \leq \varepsilon$.
In our formal results, we control the complexity of the policy class $\Pi$ in terms of its Hamming entropy.

\begin{assu}
	\label{assu:entropy}
	There exist constants $C, D \geq 0$ and $0 < \omega < 0.5$ such that, for all $ 0 < \varepsilon < 1$,  $N_{d_h}(\varepsilon, \Pi) \leq C \exp(D(\frac{1}{\varepsilon})^\omega)$.
\end{assu}

Whenever Assumption \ref{assu:entropy} holds, we can use the argument from Lemma 2 of \citet{zhou2018offline} to show 
the rate of convergence in \eqref{eq:oracle_CLT} 
holds uniformly over the whole class $\Pi$
for the oracle estimator $\tilde{\Delta}(\pi, \pi')$. The bounds below depend on the complexity of the class $\Pi$ via
\begin{equation}
\kappa(\Pi) = \int_0^1 \sqrt{\log N_{d_h}(\varepsilon^2, \Pi)} d\varepsilon,
\end{equation}
which is always finite under Assumption \ref{assu:entropy}.
\begin{exam}[The class of linear thresholding policies]
	\label{exam:linear}
	In the case of linear thresholding policies with binary actions $\abs{\acal}=2$, i.e., $\{\pi \in \Pi: \tau_\pi = \min(t: \theta^\top S_{1:t} > 0)\}$ where $\theta \in \RR^d$, we note that by \citet{haussler1995sphere}, the covering number of a policy class for single-step decision-making is bounded by $N_{L_1(\mathbb{P}_n)}(\varepsilon, \Pi_t) \leq c VC(\Pi_t)\exp^{VC(\Pi_t)} (1/\varepsilon)^{VC(\Pi_t)}$ where $VC(\Pi_t)$ is the VC dimension of $\Pi_t$, the linear thresholding policy class at time $t$, and c is some numerical constant. Thus, with a different constant $c$, $N_{L_1(\mathbb{P}_n)}(\varepsilon, \Pi_t) \leq c de^d (1/\varepsilon)^d$. By taking a Cartesian product of the covering at each timestep and with a union bound on the error incurred at each timestep, we achieve a strict upperbound on $N_{d_h}(\varepsilon, \Pi) < c d^Te^{dT} (T/\varepsilon)^{dT}$ for a (again different) constant $c$, and so $\kappa(\Pi) < c dT \log(T)$.
\end{exam}

\begin{lemm}
	\label{lemm:oracle-regret}
	Under Assumptions \ref{assu:consistency}--\ref{assu:entropy} and assuming $\abs{Y} \leq M$ for some constant $M$ almost surely, for any $\delta,\,c >0$, there exists $0 <\varepsilon_0(\delta, c) < \infty$  and universal constants $0< c_1,\, c_2 < \infty$ such that for all $\varepsilon < \varepsilon_0(\delta,c)$, 
	if we collect at least $n(\varepsilon,\delta)$ samples, with 
	\begin{align}
	n(\varepsilon, \delta) = \frac{1}{\varepsilon^2}\p{c + \sqrt{V^*}\p{c_1 \kappa(\Pi) + c_2 + \sqrt{2\log \p{\frac{1}{\delta}}}}}^2,
	\end{align} 
	where $V^* = \sup_{\pi, \pi' \in \Pi} \EE{\langle  \pi(H^{(i)}) - \pi'(H^{(i)}), \tilde{\Gamma}^{(i)} \rangle^2}$,
	then, with probability at least $1-2\delta$,
	\begin{align}
	\sup_{\pi, \pi' \in \Pi} \abs{\tilde{\Delta}(\pi,\pi') - \Delta(\pi,\pi')} \leq \varepsilon,
	\end{align}
	and, moreover, letting $\tilde{\pi} = \argmax\{\tilde{\Delta}(\pi, \bm{0}) : \pi \in \Pi\}$ be the
	policy learned by optimizing the oracle objective \eqref{eq:plugin-oracle-estimator}, we have with probability at least $1-2\delta$,
	$R(\tilde{\pi}) \leq \varepsilon$.
\end{lemm}

Our goal is to get a comparable regret bound using the feasible estimator from \eqref{eq:feasible-adr} in Algorithm \ref{alg:adr} that uses estimated nuisance components
by coupling the feasible value estimates with the oracle ones. We establish our coupling result in terms of rates of
convergence on the nuisance components, as follows.

\begin{assu}\label{assu:nuisance} We work with a sequence of problems and estimators such that $\hat{\mu}^{-q(i)}_{now,k}$, $\hat{\mu}^{-q(i)}_{next,k}$, $\hat{e}^{-q(i)}_{t,k}$, satisfy for some universal constants $C_\mu, \, C_e, \,  \kappa_\mu, \, \kappa_e$,
	\begin{align}
	&\sup_{k, t}\EE{\p{\hat{\mu}^{-q(i)}_{now,k}(S_{1:t},t) -\mu_{now,k}(S_{1:t},t)}^2}\leq C_\mu n^{-2\kappa_\mu},\\
	& \sup_{k, t}\EE{\p{\hat{\mu}^{-q(i)}_{next,k}(S_{1:t},t) -\mu_{next,k}(S_{1:t},t)}^2} \leq C_\mu n^{-2\kappa_\mu},\\
	&\sup_{k, t}\EE{\p{\frac{1}{\hat{e}^{-q(i)}_{t,k}(S_{1:t})} -\frac{1}{e_{t,k}(S_{1:t})}}^2} \leq C_e n^{-2\kappa_e}, \label{eq:ehat-rate}
	\end{align}
	and furthermore $\hat{e}^{-q(i)}_{t,k}(S_{1:t})$ is uniformly consistent, 
	\begin{align}
	\sup_{t,\,k,\, s\in \mathcal{S}^T} \abs{\hat{e}^{-q(i)}_{t,k}(s_{1:t}) - e_{t,k}(s_{1:t})} \to_p 0.
	\end{align}
\end{assu}

Moreover, motivated by the observation that treatment effects are weak relative to the available sample
size in many problems of interest, we allow for problem sequences where treatment effects can shrink with sample size $n$. In regimes where
treatment effects stay constant when the sample size grows, super-efficiency phenomena are unavoidable;
see \citet{luedtke2017faster} for a formal statement with static decision rules.
To consider other such settings, first recall the definition of $\delta_{local, k}$ defined in \eqref{eq:delta}. We further define $\delta_{local,k}^+(S_{1:t+1},t) = \mu_{now, k}(s_{1:t+1},t+1) - \mu_{next, k}(s_{1:t}, t)$.\footnote{In deterministic systems this difference is always identically zero, but in stochastic settings the two quantities will generally be different. Intriguingly related quantities (the expected temporal-difference error, and the variance of the value of next state) have been observed to play important roles in online reinforcement learning regret bounds \citep[see, e.g.,][]{zanette2019tighter} as well.}

\begin{assu}\label{assu:delta}
	For some universal constants $C_\delta, \,  \kappa_\delta$,  $C_\gamma, \,  \kappa_\gamma$,
	\begin{align}
	&\sup_{t, k} \EE{\delta_{local, k}(S_{1:t},t)^2} \leq C_\delta n^{-2\kappa_\delta}, \\
	&	\sup_{t, k} \EE{\delta_{local,k}^+(S_{1:t+1},t)^2} \leq C_\gamma n^{-2\kappa_\gamma}.
	\end{align}
\end{assu}

\begin{lemm}
	\label{lemm:cf-uniform}
	Suppose that Assumptions \ref{assu:consistency}--\ref{assu:delta} hold and assume $\abs{Y} \leq M$ for some constant $M$ almost surely. Then, for any $\delta > 0$,  there exists $0 < \varepsilon_0(\delta, \, \eta, \, T) < \infty$ such that for all $\varepsilon <\varepsilon_0(\delta, \, \eta, \, T)$, 
	with probability at least $1-3\delta$, 
	\begin{align*}
	\sup_{\pi,\pi'}\abs{\hat{\Delta}(\pi,\pi') - \tilde{\Delta}(\pi,\pi')} \leq \varepsilon,
	\end{align*}
	provided we collect at least $n_0(\varepsilon, \delta)$ samples, where
	\begin{align*}
	n_0(\varepsilon, \, \delta) = \p{C(\delta)  KT^2 \varepsilon^{-1}}^{1/\min\cb{1/2 + \kappa_e, \, 1/2 + \kappa_\mu, \, \kappa_e + \kappa_\mu, \, \kappa_e + \kappa_\delta, \, \kappa_e + \kappa_\gamma}},
	\end{align*}
	where $C(\delta)$ only depends on the constants used in Assumptions \ref{assu:overlap}, \ref{assu:nuisance} and \ref{assu:delta}.
\end{lemm}

Combining the above with Lemma \ref{lemm:oracle-regret}, we immediately have the following finite-sample bound for the regret on the feasible estimator.

\begin{theo}
	\label{theo:main-regret}
	Let \smash{$\hat{\pi} = \argmax\{\hat{\Delta}(\pi, \bm{0}) : \pi \in \Pi\}$} be the
	policy learned by optimizing the feasible objective \eqref{eq:feasible-adr}. Suppose Assumptions \ref{assu:consistency}--\ref{assu:delta} and assume $\abs{Y} \leq M$ for some constant $M$ almost surely. Then, for any $\delta>0$, there exist $0 < \varepsilon_0(\delta, \, \eta, \, T)< \infty$ such that the following statement holds for all $\varepsilon <\varepsilon_0(\delta, \, \eta, \, T)$:
	If we collect at least $n(\varepsilon,\delta)$ samples, with 
	\begin{equation}
	\begin{split}
	n(\varepsilon, \delta) = &\max\bigg\{\frac{1}{\varepsilon^2}\p{c + \sqrt{V^*}\p{c_1 \kappa(\Pi) + c_2 + \sqrt{2\log \p{\frac{1}{\delta}}}}}^2, n_0(\varepsilon, \, \delta) \bigg\},
	\end{split}
	\end{equation}
	$V^* = \sup_{\pi, \pi' \in \Pi} \EE{\langle  \pi(H^{(i)}) - \pi'(H^{(i)}), \tilde{\Gamma}^{(i)} \rangle^2}$,
	and $n_0(\varepsilon,\delta)$ as defined in Lemma \ref{lemm:cf-uniform},
	then with probability at least $1-5\delta$
	\begin{align*}
	\sup_{\pi, \pi' \in \Pi} \abs{\hat{\Delta}(\pi,\pi') - \Delta(\pi,\pi')} \leq 2\varepsilon,
	\end{align*}
	and in particular \smash{$R(\hat{\pi}) \leq 2\varepsilon$}.
\end{theo}

We obtain the following corollary if we assume specific learning rates on the nuisance components and the signal strength.

\begin{coro}
	\label{coro:regret}
	Assume $\kappa_\mu > 0,\, \kappa_e > 0, \,\kappa_e + \kappa_\mu > \frac{1}{2}, \, \kappa_e + \kappa_\gamma > \frac{1}{2}, \,\kappa_{e} + \kappa_\delta > \frac{1}{2}$. Suppose Assumptions \ref{assu:consistency}--\ref{assu:delta} hold and assume $\abs{Y} \leq M$ for some constant $M$ almost surely. Then, for any $\delta >0$, there exist $0 <\varepsilon_0(\delta, \, \eta, \, T) < \infty$  such that the following holds for all $\varepsilon < \varepsilon_0(\delta, \, \eta, \, T)$:
	If we collect at least $n(\varepsilon,\delta)$ samples, with 
	\begin{align}
	\label{eq:coro_regret}
	n(\varepsilon, \delta) = \frac{1}{\varepsilon^2}\p{c + \sqrt{V^*}\p{c_1 \kappa(\Pi) + c_2 + \sqrt{2\log \p{\frac{1}{\delta}}}}}^2,
	\end{align}
	and $V^* = \sup_{\pi, \pi' \in \Pi} \EE{\langle  \pi(H^{(i)}) - \pi'(H^{(i)}), \tilde{\Gamma}^{(i)} \rangle^2}$, then
	with probability at least $1-5\delta$
	\begin{align*}
	\sup_{\pi, \pi' \in \Pi} \abs{\hat{\Delta}(\pi,\pi') - \Delta(\pi,\pi')} \leq 2\varepsilon,
	\end{align*}
	and in particular \smash{$R(\hat{\pi}) \leq 2\varepsilon$}.
\end{coro}

Our result above can be interpreted in several different regimes. First, we note that we can reach the optimal sample complexity $n \sim \varepsilon^{-2}$ if either (a) the treatment propensities $e_{t,k}$ are known and we can consistently estimate $\mu_{now,k}$ and $\mu_{next,k}$; or (b) the signal size of the advantages is null (i.e., $\mu_{now,k}(S_{1:t},t) - \mu_{next,k}(S_{1:t},t) = 0$) or is weak (in the sense that $\kappa_\delta > 0$ and $e_{t,0}$ can be learned at a rate such that $\kappa_\delta + \kappa_{e} > 1/2$, etc.), and similarly the stochastic fluctuations are weak (in the sense that $\kappa_\gamma > 0$ and $e_{t,0}$ can be learned at a rate such that $\kappa_\gamma + \kappa_{e} > 1/2$, etc.), and we can consistently estimate $\mu_{now,k}$, $\mu_{next,k}$ and $e_{t,k}$ such that  $\kappa_e + \kappa_\mu>1/2$.

Conversely, if the treatment effects are of a fixed size (i.e., $\kappa_\delta = 0$), and we don't know the treatment propensities $e_{t,k}$ a priori, then we pay a price for not being robust to the change of measure from Lemma \ref{lemm:murphy-decomp} to Lemma \ref{lemm:is-decomp}, and we no longer achieve the optimal rate. The terms that hurt us are due to error terms that decay as $n^{-(\kappa_e + \kappa_\delta)}$
which arises from the interaction of how we use inverse propensity weighting for the treatment starting probabilities
and the signal size of the advantages, and ones that decays as $n^{-(\kappa_e + \kappa_\gamma)}$ which are similarly due to stochastic fluctuations in the
value of starting treatment. If advantages are small, this won't matter for smaller target error rates $\varepsilon$, but
requires a bigger sample size when we aim for very small $\varepsilon$.

\section{ADR with a Terminal State}
\label{sec:death}

So far, we have focused our analysis on when-to-stop problems in settings characterized by overlap
(Assumption \ref{assu:overlap}), i.e., where the sampling policy can start treatment in any state with
a non-zero policy, and have assumed that we want to learn a regular policy in the sense of
Definition \ref{def:regular_policy}, i.e., one that never stops prescribing treatment once it has started to do so.
In many applications of interest, however, a patient may enter a terminal state in which treatment becomes
impossible---for example, a patient may leave the study or die. The existence of such a terminal state
contradicts the assumptions made above: There is no overlap in the terminal state (because treatment can
never start there), and a policy that respects the terminal state may not be regular (because the policy must
stop prescribing treatment once the patient enters the terminal state).

The goal of this section is to briefly discuss methodological extensions to ADR that are required in the presence
of a terminal state. Algorithm~\ref{alg:adr-terminal} provides pseudocode for our ADR policy optimization approach with terminal states. 

To do so, we start by adapting Definition \ref{def:regular_policy} and
Assumption \ref{assu:overlap} to this setting.

\begin{defi}[Terminal state]
	\label{defi:terminal}
	A state $\Phi \in \set$ is terminal if, whenever $S_t = \Phi$ , then also
	$S_{t'} = \Phi$ for all $t' > t$. Furthermore, we assume that once a patient enters a terminal state,
	we can assess their final outcome, i.e., there exists a set of known functions\footnote{For example,
		if $Y$ is survival time, then one can use $H_t(S_{1:t}) = \sup\cb{t' : S_{t'} \neq \Phi, \, t'\leq t}$.}
	$H_t$ such that $Y = H_t(S_{1:t})$ whenever $S_{t+1} = \Phi$.
\end{defi}

\setcounterref{deficmp}{def:regular_policy}
\addtocounter{deficmp}{-1}

\begin{deficmp}[Regular policy with terminal state]
	\label{def:regular_policy_term}
	A regular when-to-treat policy $\pi$ that respects the terminal state $\Phi$
	is determined by an $\ff_t$-measurable stopping time $\tau_\pi$
	and an associated $\ff_{\tau_\pi}$-measurable decision variable $W_\pi \in \cb{1, \, \ldots, \, K}$
	as follows:\footnote{Note that the policies specified here are still when-to-start policies, i.e.,
		if they have actually started treatment then they never stop (even if the patient enters a terminal
		state). One could also choose to make $\pi$ stop treatment once the patient enters a terminal state.
		However, from a statistical perspective, this makes no difference: All that matters is that the standard
		of care is a deterministic function of state once treatment has started.}
	For each time $t = 1, \, \ldots,\, T$,  if $S_t = \Phi$ then $\pi_t(S_{1:t}, \, A_{1:(t-1)}) = 0$. Otherwise, if $A_{t - 1} \neq 0$ then $\pi_t(S_{1:t}, \, A_{1:(t-1)}) = A_{t-1}$, else if $t \geq \tau_\pi$ then $\pi_t(S_{1:t}, \, A_{1:(t-1)}) = W_\pi$. If none of the above conditions apply, then $\pi_t(S_{1:t}, \, A_{1:(t-1)}) = 0$.
\end{deficmp}

\setcounterref{assucmp}{assu:overlap}
\addtocounter{assucmp}{-1}

\begin{assucmp}[Overlap with terminal state]\label{assu:overlap_term}
	There are constants $\eta, \, \eta_0 > 0$ as well as a terminal state $\Phi$ such that, for all $t=1,\cdots, T$
	and $s_{1:t} \in \mathcal{S}^t$, the following hold.
	If $s_t = \Phi$, then $e_{t,a}(s_{1:t}) = 0$ for all $a \in \mathcal{A} \setminus \cb{0}$ and  $ e_{t,0}(s_{1:t}) = 1$; 
	else, $e_{t,a}(s_{1:t}) > \eta/T$ for all $a \in \mathcal{A} \setminus \cb{0}$ and $ e_{t,0}(s_{1:t}) > 1 - \eta_0/T$.
\end{assucmp}

In the presence of a terminal state, the main modification we need to make to ADR is that the
conditional expectation $\mu_{next, k}(s_{1:t}, t)$ as defined in \eqref{eq:mu_now_next} no longer matches
the $Q$-function that arises in \eqref{eq:murphy} in the proof of Lemma \ref{lemm:murphy-decomp}, and so we need to adapt
our statement of this result. The proof of the following lemma is included in the Appendix.

\begin{lemm}\label{lemm:murphy-decomp-term}
	Under Assumptions \ref{assu:consistency} and \ref{assu:ignorability}, let $\Phi$ be the terminal state, and let $\pi$ be a regular when-to-treat
	policy that respects $\Phi$ in the sense of Definition \ref{def:regular_policy_term}.
	Then
	\begin{equation}
	\label{eq:murphy-decomp-term}
	\Delta(\pi, \bm{0}) = \EE[\bm{0}] {\sum_{t=\tau_\pi}^T \mathbbm{1}_{S_t \neq \Phi}\p{\mu_{now, W_\pi}(S_{1:t}, t) - \mu^\Phi_{next, W_\pi}(S_{1:t}, t)}},
	\end{equation}
	where $\mathbb{E}_{\bm{0}}$ samples trajectories under a never-treating policy,
	$\tau_\pi$ is the time at which $\pi$ starts treating
	and $W_\pi$ is the treatment chosen at that time, and
	\begin{equation}
	\label{eq:mu-next-term}
	\begin{split}
	& \mu^\Phi_{next, k}(S_{1:t}, t) \\
	&\ \
	= \PP{S_{t + 1} \neq \Phi \cond S_{1:t}, \, A_{1:t} = \bm{0}} 
	\EE{\mu_{now, k}\p{S_{1:t+1}, \, t+1} \! \cond S_{1:t}, \, A_{1:t} = \bm{0}, \, S_{t+1} \neq \Phi} \\
	&\ \ \ \ 
	+ \PP{S_{t + 1} = \Phi \cond S_{1:t}, \, A_{1:t} = \bm{0}} H_t(S_{1:t}).
	\end{split}
	\end{equation}
\end{lemm}

Next, the following result is a direct consequence of Lemma \ref{lemm:murphy-decomp-term}; its proof is a direct
analogue to that of Lemma \ref{lemm:is-decomp} and thus omitted.

\begin{lemm} \label{lemm:is-decomp-term} In the setting of Lemma \ref{lemm:murphy-decomp-term} and
	under Assumptions \ref{assu:consistency}, \ref{assu:ignorability} and \ref{assu:overlap_term},
	\begin{align}\label{eq:is-adv-decomp-term}
	\Delta(\pi, \bm{0}) = \EE {\sum_{t=\tau_\pi}^T \mathbbm{1}_{S_t \neq \Phi}\frac{\mathbbm{1}_{A_{1:t-1}=0}}{\prod_{t'=1}^{t-1}e_{t',0}(S_{1:{t'}})} \p{\mu_{now, W_\pi}(S_{1:t}, t) - \mu^{\Phi}_{next, W_\pi}(S_{1:t}, t)}}.
	\end{align}
\end{lemm}

We detail a candidate estimator based on this result as Algorithm \ref{alg:adr-terminal}.
For notational convenience, we denote terminating probabilities by $\rho(S_{1:t}) = \PP{S_{t + 1} = \Phi \cond S_{1:t}, \, A_{1:t} = \bm{0}}$,
and write $U(S_{1:t}. \Phi) = 	\EE{\frac{\mathbbm{1}_{A_{t+1}=k}}{e_{t+1,k}(S_{1:t+1})}Y  \cond S_{1:t}, \, A_{1:t} = \bm{0}, \, S_{t+1} \neq \Phi}$. 
We note that the proposed ADR estimator extended to terminal states is robust towards errors in estimating the regression outcome functions $\mu_{now,k}$ and $U(\cdot, \Phi))$ but we do not correct for the estimation bias in estimating the terminating probabilities $\rho(\cdot)$. We leave it to future work to develop robust methods that are also robust against terminating probability estimates.

Finally, in analogy to \eqref{eq:rewritenext}, it is convenient to re-express $\mu_{next,k}(S_{1:t}.t)^{\Phi}$ via inverse-propensity
weighting for purpose of estimating it,
\begin{equation}
\begin{split}
& \mu^\Phi_{next, k}(S_{1:t}, t) \\
&\ \ \ \
= \PP{S_{t + 1} \neq \Phi \cond S_{1:t}, \, A_{1:t} = \bm{0}} 
\EE{\frac{\mathbbm{1}_{A_{t+1}=k}}{e_{t+1,k}(S_{1:t+1})}Y  \cond S_{1:t}, \, A_{1:t} = \bm{0}, \, S_{t+1} \neq\Phi} \\
&\ \ \ \ \ \ \ \
+ \PP{S_{t + 1} =\Phi \cond S_{1:t}, \, A_{1:t} = \bm{0}} 
H_t(S_{1:t}).
\end{split}
\end{equation}
Furthermore, a weighted regression expression analogous to \eqref{eq:weighted_reg} also holds.

\RestyleAlgo{boxruled}
\LinesNumbered
\begin{algorithm}[t]
	\caption{\textbf{Advantage Doubly Robust Estimator with Terminal State}\label{alg:adr-terminal}}
	Estimate the outcome models $\mu_{now,k}(\cdot)$,
	$U(\cdot, \Phi)$, terminating propensities $\rho(\cdot)$ as well as treatment propensities $e_{t,a}(s_{1:t})$ with cross fitting using any supervised learning method tuned for prediction accuracy.
	
	Given these nuisance component estimates, we construct value estimates 
	\begin{equation}
	\label{eq:feasible-adr-terminal}
	\hat{\Delta}(\pi, \bm{0})
	= \frac{1}{n} \sum_{i=1}^{n} \sum_{t=1}^T \mathbbm{1}_{S_t\neq \Phi} \mathbbm{1}_{t \geq \tau_\pi^{(i)}} \frac{\mathbbm{1}_{A_{1:t-1}^{(i)} = 0}}{\prod_{t'=1}^{t-1}\hat{e}^{-q(i)}_{t',0}(S_{1:{t'}}^{(i)})}\hat{\Psi}^{\Phi}_{t,W_\pi}(S_{1:t}^{(i)})
	\end{equation}
	for each policy $\pi \in \Pi$,  where the relevant doubly robust score is
	\begin{equation}
	\begin{split}
	\hat{\Psi}^{\Phi}_{t,k}(S_{1:t}^{(i)}) &= \hat{\mu}^{-q(i)}_{now,k}(S_{1:t}^{(i)},t) - {\hat{\mu}^\Phi_{next,k}}(S_{1:t}^{(i)},t)^{-q(i)} \\
	&\ \ \ \ \ + \mathbbm{1}_{A_t^{(i)}=k} \frac{Y^{(i)} - \hat{\mu}^{-q(i)}_{now, k}(S_{1:t}^{(i)},t)}{\hat{e}^{-q(i)}_{t,k}(S_{1:t}^{(i)})} \\
	&\ \ \ \ \ - \mathbbm{1}_{A_t^{(i)}=0} \mathbbm{1}_{A_{t+1}^{(i)}=k} \frac{Y^{(i)} - \hat{U}^{-q(i)}(S_{1:t}^{(i)},\Phi)}{\hat{e}^{-q(i)}_{t,0}(S_{1:t}^{(i)})\hat{e}^{-q(i)}_{t+1,k}(S_{1:t+1}^{(i)})},
	\end{split}
	\end{equation}
	and ${\hat{\mu}^\Phi_{next,k}}(S_{1:t}^{(i)},t)^{-q(i)} = (1-\hat{\rho}^{-q(i)}(S_{1:t}^{(i)})) \hat{U}^{-q(i)}(S_{1:t}^{(i)}, \Phi) + \hat{\rho}^{-q(i)}(S_{1:t}^{(i)}) H_t(S_{1:t}^{(i)})$
	
	Learn the optimal policy by setting $\hat{\pi} = \argmax_{\pi \in \Pi} \hat{\Delta}(\pi, \bm{0})$.
\end{algorithm}

\section{Experiments}
\label{sec:experiments}

In order to assess the practical performance of our proposed method, we consider two different simulation studies. In the first simulation, we consider the optimal stopping case in which the treatment decision is binary and the treatment assignment propensities are not known a-priori. In the second simulation, we want to learn when to start which treatment. There are multiple treatment options and patients can be censored due to death, and the data are generated from a randomized control trial with known treatment assignment propensities. This second study helps to capture settings motivated by clinical trials. 

In both settings we consider linear thresholding policy rules for simplicity and due to their interpretability. In our implementation, we use the normalized variant of the IPW estimator $\hat{V}_\pi^{\textrm{WIPW}}$ as presented in Section \ref{sec:existingwork}. For the first setup that doesn't involve survival censoring, we use a correspondingly normalized ADR estimator $\hat{\Delta}^{\textrm{W}}$ in Step 2 of Algorithm \ref{alg:adr}:
\begin{equation*}
\begin{split}
\hat{\Delta}^{\textrm{W}}(\pi, \bm{0})
&=  \sum_{t=1}^T  \frac{\sum_{i=1}^{n}  \frac{\mathbbm{1}_{A_{1:t-1}^{(i)} = 0}}{\prod_{t'=1}^{t-1}\hat{e}^{-q(i)}_{t',0}(S_{1:{t'}}^{(i)})}\mathbbm{1}_{t \geq \tau_\pi^{(i)}} \p{\hat{\mu}^{-q(i)}_{now,W_\pi}(S_{1:t}^{(i)},t) - \hat{\mu}^{-q(i)}_{next,W_\pi}(S_{1:t}^{(i)}, t)}}{\sum_{i=1}^{n}  \frac{\mathbbm{1}_{A_{1:t-1}^{(i)} = 0}}{\prod_{t'=1}^{t-1}\hat{e}^{-q(i)}_{t',0}(S_{1:{t'}}^{(i)})}}\\
&+  \sum_{t=1}^T  \frac{\sum_{i=1}^{n}  \frac{\mathbbm{1}_{A_{1:t-1}^{(i)} = 0}\mathbbm{1}_{A_t^{(i)}=W_\pi}}{\prod_{t'=1}^{t-1}\hat{e}^{-q(i)}_{t',0}(S_{1:{t'}}^{(i)})\hat{e}^{-q(i)}_{t,W_\pi}(S_{1:t}^{(i)})}\mathbbm{1}_{t \geq \tau_\pi^{(i)}} \p{Y^{(i)} - \hat{\mu}^{-q(i)}_{now, W_\pi}(S_{1:t}^{(i)},t)}}{\sum_{i=1}^{n}  \frac{\mathbbm{1}_{A_{1:t-1}^{(i)} = 0}\mathbbm{1}_{A_t^{(i)}=W_\pi}}{\prod_{t'=1}^{t-1}\hat{e}^{-q(i)}_{t',0}(S_{1:{t'}})\hat{e}^{-q(i)}_{t,W_\pi}(S_{1:t}^{(i)})}\mathbbm{1}_{t \geq \tau_\pi^{(i)}} + \sum_{i=1}^{n}  \frac{\mathbbm{1}_{A_{1:t-1}^{(i)} = 0}}{\prod_{t'=1}^{t-1}\hat{e}^{-q(i)}_{t',0}(S_{1:{t'}})}(1-\mathbbm{1}_{t \geq \tau_\pi^{(i)}})}\\
&-  \sum_{t=1}^T  \frac{\sum_{i=1}^{n}  \frac{\mathbbm{1}_{A_{1:t}^{(i)} = 0}\mathbbm{1}_{A_{t+1}^{(i)}=W_\pi}}{\prod_{t'=1}^{t}\hat{e}^{-q(i)}_{t',0}(S_{1:{t'}})\hat{e}^{-q(i)}_{t+1,W_\pi}(S_{1:t+1}^{(i)})}\mathbbm{1}_{t \geq \tau_\pi^{(i)}} \p{Y^{(i)} - \hat{\mu}^{-q(i)}_{next, W_\pi}(S_{1:t}^{(i)},t)}}{  \frac{\mathbbm{1}_{A_{1:t}^{(i)} = 0}\mathbbm{1}_{A_{t+1}^{(i)}=W_\pi}}{\prod_{t'=1}^{t}\hat{e}^{-q(i)}_{t',0}(S_{1:{t'}})\hat{e}^{-q(i)}_{t+1,W_\pi}(S_{1:t+1}^{(i)})}\mathbbm{1}_{t \geq \tau_\pi^{(i)}} +  \frac{\mathbbm{1}_{A_{1:t-1}^{(i)} = 0}}{\prod_{t'=1}^{t-1}\hat{e}^{-q(i)}_{t',0}(S_{1:{t'}})}(1-\mathbbm{1}_{t \geq \tau_\pi^{(i)}}) 
}.
\end{split}
\end{equation*}
For simplicity, we will refer to them as the IPW (baseline) and ADR (our estimator) respectively in this section. We have a similar weighted form for the ADR estimator with terminal states that we use in experiments in the second simulation study. We include that in the Appendix. 

In addition to IPW, we also consider fitted-Q iteration as a baseline method for policy learning.
The variant of fitted-Q iteration we implement follows the Batch Q-learning algorithm as described in
\citep{murphy2005generalization} for solving the optimal $Q$ function at each timestep:
At each $t=T, T-1, \cdots, 1$, we solve\footnote{For this purpose, we estimate propensities and conditional
	response surfaces using regression forests as implemented in \texttt{grf} \citep{athey2019generalized}. For tractability,
	we do not consider history when learning these regression; rather, we only use current state as covariates in each time step.}
\begin{align}
\label{eq:FQ}
\hat{Q}^*_t(\cdot, \cdot) = \argmin_{Q_t} \frac{1}{n} \sum_{i=1}^n \p{\max_{a_{t+1}} \hat{Q}^*_{t+1}(S_{1:(t+1)}^{(i)}, \{A_{1:t}^{(i)}, a_{t+1}\}) - Q_t(S_{1:t}^{(i)}, A_{1:t}^{(i)})}^2,
\end{align}
where we let $\hat{Q}^*_{T+1} = Y^{(i)}$.
We note that fitted-Q iteration is an iterative backwards-regression based algorithm targeted at learning the optimal unrestricted
policy, whereas our goal is to learn the best in-class policy given a user-defined policy class. However, while fitted-Q aims to perform
a different task then us, it is still of interest to compare the regret achieved by both methods. We use the shorthand \textit{Q-Opt} to refer to this method. Another variant of the fitted-Q method evaluates a given policy $\pi$ instead of learning the best policy \citep[see, e.g.,][]{le2019batch}. Instead of taking the \textit{max} operator above, $a_{t+1}$ is chosen according to the policy $\pi$. We call the latter fitted-Q for evaluation and use \textit{Q-Eval} as a shorthand accordingly. All experiments can be replicated here: \url{https://github.com/xnie/adr}.

\subsection{Binay Treatment Choices in an Observational Study}
\label{sec:binary}

Our first simulation is motivated by a setting where we
track a health metric and get a reward if the health metric is above a threshold at $T = 10$. The treatment provides
a positive nudge to the health metric at a cost. We start with treatment on, and need to choose when to
stop to minimize cost while trying to keep the health metric stay above the threshold. The data generating process is as follows:
\begin{align*}
X_1 \sim \mathcal{N}(0,\sigma^2), \,\,\,\,\,\,\, &X_{t+1}\cond X_t, A_t \sim \mathbbm{1}_{X_t\geq -0.5}\mathcal{N}\p{X_t+ \frac{1}{1+e^{0.3 X_t}} A_t, \frac{\sigma^2}{2T}} + \mathbbm{1}_{X_t < -0.5}X_t\\
&S_t\cond X_t \sim \mathcal{N}(X_t, \nu^2), \,\,\,\, Y = \beta\mathbbm{1}_{S_{T+1} > 0} - \frac{1}{T}\sum_{t=1}^T A_t,
\end{align*}
with the stopping action $A_t \cond X_t \sim Bernoulli(1-1/(1+e^{-(X_t-1.5)} - e^{-(t-3)}))$. We note that $Y$ is the final outcome we'd like to maximize. We also do not assume Markovian structure and only get to observe $S_t$, which is a noisy version of the underlying state $X_t$.

In our implementation, both the propensity and outcome regressions only use the current state and action information as opposed to the full history even though the underlying dynamic is not Markovian.\footnote{In other words, neither out propensity models nor conditional response models are well specified because we do not use covariates that capture lagged states. Thus, this setting can be seen as a test case for the value of the robust scoring method in ADR.} We parameterize the policy class of interest by $[\theta_1, \theta_2, \theta_3]$ and define each policy to be a linear thresholding rule $\theta_1 S_t \geq \theta_2 t +  \theta_3$ such that whenever this holds, we stop the treatment. We then perform a grid search over a range of values for the policy parameters, with the grid specified in Appendix \ref{append:plots}. 

\begin{figure}[!htbp]
	\centering
	\begin{tabular}{cc}
		\includegraphics[width=0.42\columnwidth]{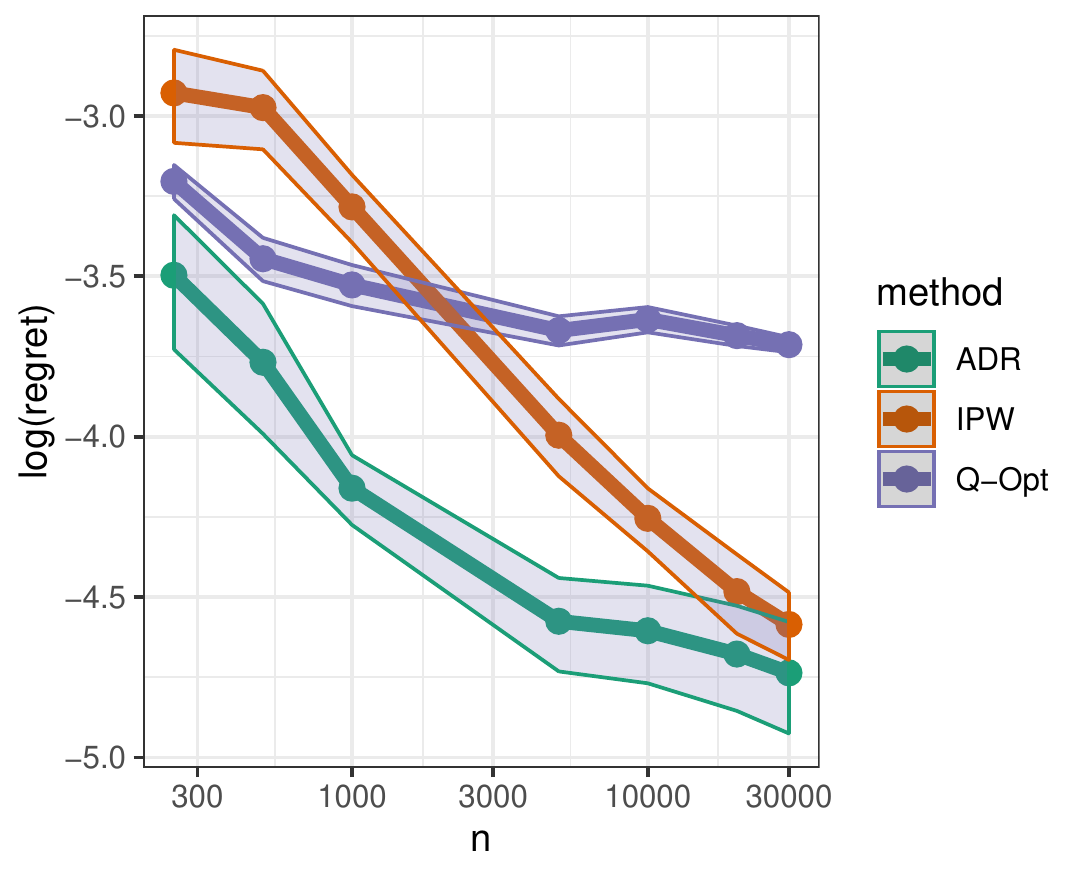} &
		\includegraphics[width=0.42\columnwidth]{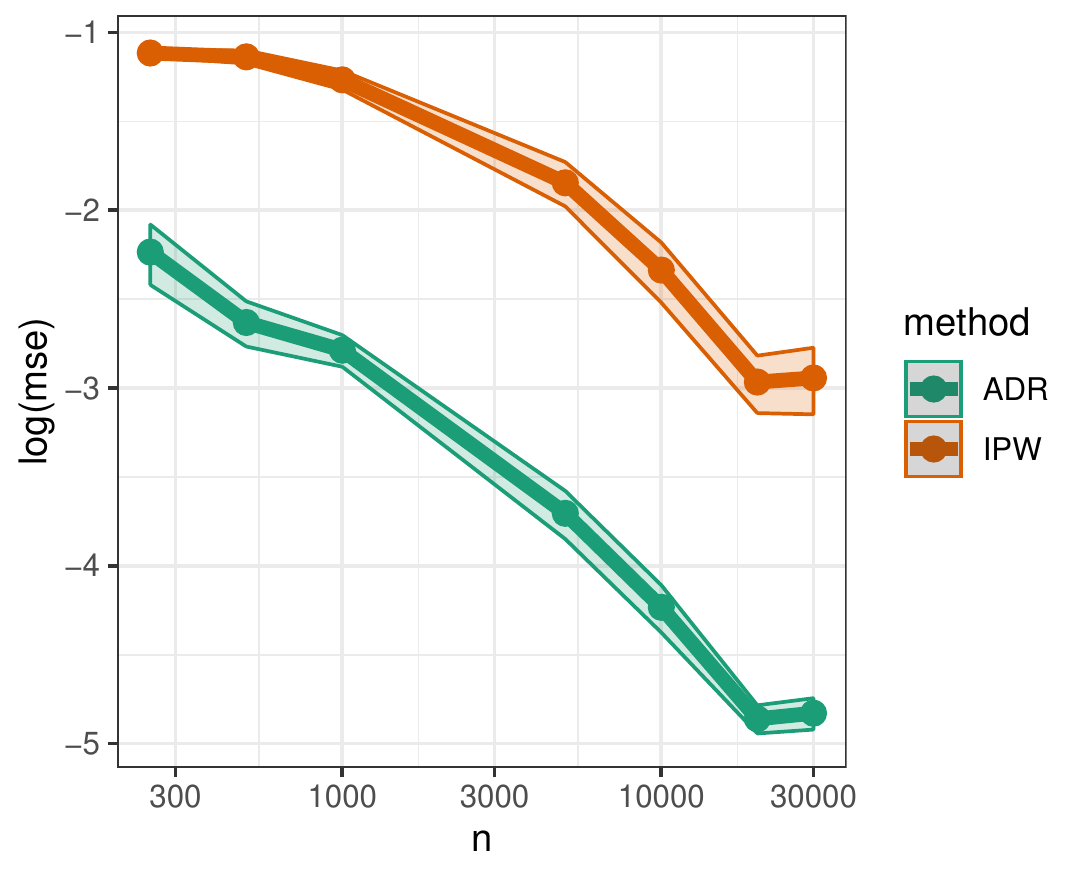} 
	\end{tabular}
	\vspace{-0.7\baselineskip}
	\caption{
		We compare the performance of ADR in comparison to IPW and \textit{Q-Opt} using $\sigma=1$, $\beta=1$ and $\nu=0.5$ in the binary treatment setup. We plot the regret (left figure) relative to the best in-class policy and the average mean-squared error (right figure) of the value estimates for policies in the same policy class across all policies (both in log-scale). The shaded regions are standard error bars. In the mean-squared error (MSE) plot, the MSE for each policy is computed against an oracle evaluation using Monte-Carlo rollouts using the underlying transition dynamics averaged across 20000 runs. Both the regret and MSE results are averaged acros 50 runs. The $x$-axis shows the number of offline trajectories we generate in the observational data. 
		\vspace{-\baselineskip}}
	\label{fig:binary}
\end{figure}

For each of the parameter combinations, we run ADR and baseline IPW to estimate the value of the corresponding policy. The average mean-squared error (MSE) of each of the policy values across all policies in the policy class is then computed against an oracle evaluation by using a Monte-Carlo rollout of the policy using the underlying transition dynamics averaged across 20000 times. We vary $\beta$, $\sigma$, and the observation noise $\nu$ and compute the regret and the mean-squared error of policy value estimates (averaged across all policies in the policy class).

ADR shows a clear advantage in both regret and learning the correct value of policies across varying values of $\sigma,\beta$ and $\nu$. We present the tables of raw results for  in Table \ref{table:setup2-beta-0_5}-\ref{table:setup2-beta-1} in Appendix \ref{append:plots}. We present one representative illustration in Figure \ref{fig:binary}, where we have used $\sigma=1, \beta=1$ and $\nu=0.5$. We compare the performance of ADR against IPW and \textit{Q-Opt} with varying numbers of offline trajectories.  IPW and ADR first evaluate the values of the policies in the policy class, and so we plot the MSE of their policy estimates averaged across all policies in the policy class in the right plot; it is not applicable for \textit{Q-Opt} which seeks to learn the optimal policy directly.

\subsection{Multiple Treatment Choices}
\label{sec:multiple}
In the second setup, we consider multiple treatment choices.
Our design here is motivated by a healthcare setting where, once a doctor starts treatment,
they can choose between a more effective but more invasive treatment with strong side effects,
or a less effective but less invasive treatment.
More specifically, imagine a cancer patient's state at time $t$ is modeled by $X_t$, $Y_t$ and $Z$, where $X_t$ is the general health state, $Y_t$ is the state of a tumor, and $Z$ is not time-dependent but models the category of the patients for which lifespan differs. In particular, if $Z=0$, a patient always dies immediately; if $Z = 1$, a patient always survives until the end of a trial; if $Z=2$, the patient's lifespan has a strong dependency on $Y_t$, which we detail below. There are two treatment choices, one non-invasive ($A_t=1$) and one invasive ($A_t=2$). The non-invasive option lessens the severity of the  tumor, and the invasive option completely removes the tumor, but exacerbates a patient's general health conditions. The final outcome is denoted by $R$, which is the lifetime of a patient, and we seek a policy $\pi$ that maximizes $\EE[\pi]{R}$. We consider horizon $T=10$. The data generating process is as follows: 
\begin{align}
&X_1 \sim Exp(1) \ \ \ \ 
Y_1 \sim 0.5 Exp(3)\ \ \ \ 
Z \sim Multinomial(0.3, \, 0.3, \, 0.4) \ \ \ \
L_1 = 1\nonumber \\
&Z=1: L_{t+1} = 0, \ \ \ \ \ Z=2: L_{t+1} = 1 \nonumber\\
&Z=3: L_{t+1} = 0 \ \textrm{if } L_t = 0\textrm{;}\\  
&\textrm{otherwise, }L_{t+1} \sim Bernoulli(\mathbbm{1}_{Y_t \leq 5}\exp(-0.02Y_t) + \mathbbm{1}_{5 < Y_t \leq 14}\exp(-0.06Y_t)) \nonumber\\
&A_t=0: \, X_{t+1} = \abs{X_t + \sigma_t} \ \ \ \ 
Y_{t+1} = \abs{Y_t + 0.5X_t + \sigma_t}\ \ \ \ \nonumber\\
&A_t=1: \, X_{t+1} = \abs{X_t + \sigma_t}\ \ \ \ 
Y_{t+1} = \abs{0.5 Y_t + \sigma_t}\ \ \ \ \nonumber\\
&A_t=2: \, X_{t+1} = X_t + \abs{\max(X_t^2, 1.5X_t) + \sigma_t - X_t}  \ \ \ \ 
Y_{t+1} = 0\ \ \ \ \nonumber\\
&X_t'  = \max\p{0,\min\p{X_{max}, X_t + \nu}}, \ \ \ \ Y_t'  =   \max\p{0,\min\p{Y_{max}, Y_t + \nu}} \nonumber\\
&R = \min \{t:  L_t = 0 \} - 1,\ \ \ \ \ X_{max}=10, \ \ \ \ Y_{max}=16, \ \ \ \ \sigma_t \sim \mathcal{N}(0,0.25), \ \ \ \ \ \nu \sim \mathcal{N}(0, \sigma^2)\nonumber
\end{align}
where $L_t$ is an indicator for whether the patient is alive at time $t$.

In this setting, the treatment assignment mechanism is based on sequential randomization in the data such that there are roughly equal number of trajectories that start treating at each time with either treatment option. Note that the states we observe is $X_t'$ and $Y_t'$, which is the original states added with noise, making our setup non-Markovian. We consider the following linear thresholding class: $\theta_1 X_t' + \theta_2 Y_t' + \theta_3 t \geq \theta_4$ is the region in which we start treatment. If in addition, $\theta_5 X_t' + \theta_6 Y_t' + \theta_7 t \geq \theta_8$, we use the invasive treatment and otherwise, use the non-invasive treatment. We search over the eight parameters in the policy class with a grid search, with details in Appendix \ref{append:plots}.

We compare running the ADR policy optimization procedure (as shown in Section \ref{sec:proposal}) against IPW and \textit{Q-Opt}. Like the binary-action setup, we again estimate the oracle value of all policies in the policy class with Monte-Carlo rollouts averaged across 20000 times. 

In Figure \ref{fig:multiple}, we see that for both the best value learned and the average mean-squared error, ADR outperforms IPW. We also include the complete set of results with varying noise parameter $\sigma$ in Table \ref{table:setup1} in Appendix \ref{append:plots}.
Interestingly, we see that in very large samples \textit{Q-Opt} becomes competitive with ADR. One possible explanation for this is
that ADR is only allowed to use linear thresholding policies whereas \textit{Q-Opt} learns over arbitrary policies---and,
in large samples, the increased expressivity of \textit{Q-Opt} may become helpful.

\begin{figure}[!htbp]
	\centering
	\begin{center}
		\begin{tabular}{cc}
			\includegraphics[height=0.42\textwidth, trim = 4mm 0mm 8mm 0mm]{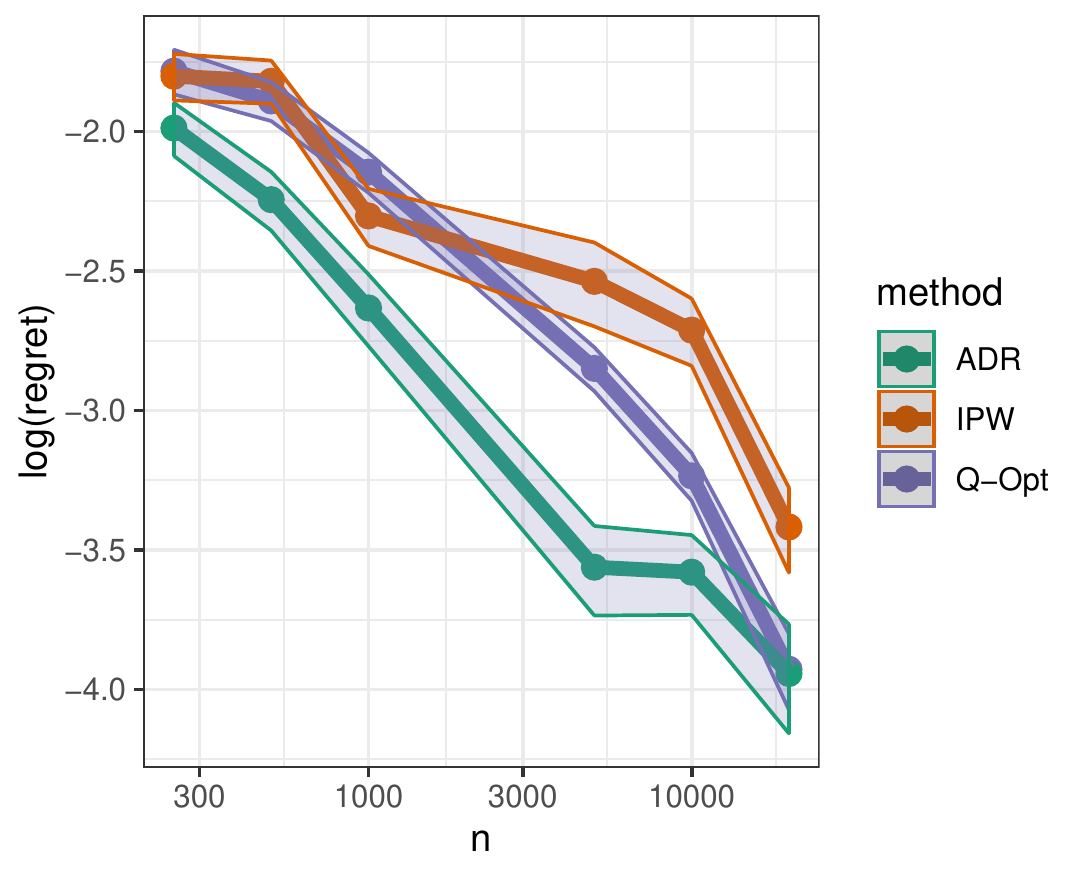}&
			\includegraphics[height=0.42\textwidth]{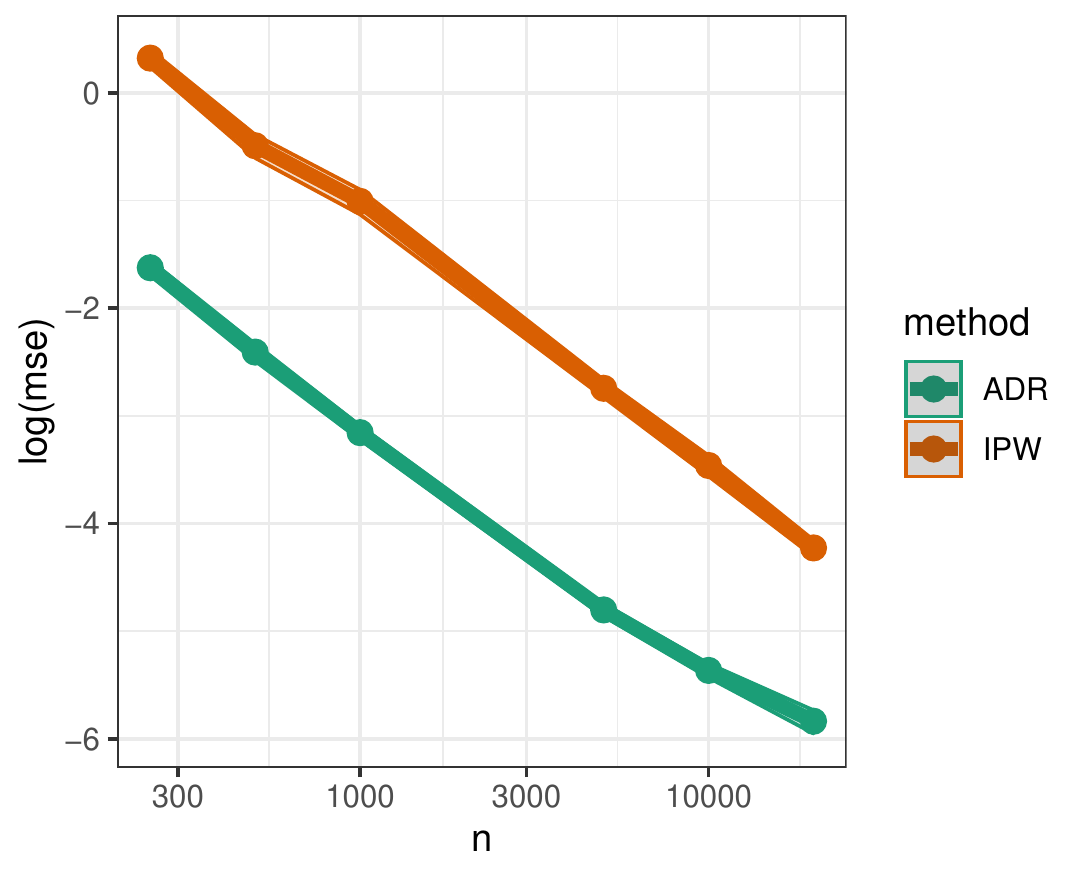}  
		\end{tabular}
		\caption{We compare the performance of ADR in comparison to IPW and \textit{Q-Opt} in the multiple treatment setup. The plot shows results for $\sigma = 1$. We plot the regret (left figure) relative to the best in-class policy and the average mean-squared error (right figure) of the value estimates for policies in the same policy class across all policies (both in log-scale). The shaded regions In the mean-squared error (MSE) plot, the MSE for each policy is computed against an oracle evaluation using Monte-Carlo rollouts under the underlying transition dynamics averaged across 20000 runs. Both the regret and MSE results are averaged across 50 runs. The $x$-axis shows the number of offline trajectories we generate in the observational data.}
		\label{fig:multiple}
	\end{center}
\end{figure}

\subsection{Policy Learning vs Policy Evaluation}

Throughout this paper, we have focused on ADR as a method for policy learning, and have emphasized
that ADR is well suited to policy learning by empirical maximization because it can
evaluate any policy in the policy class $\Pi$ using a single set of universal scores as in \eqref{eq:feasible-adr}. In contrast, standard doubly robust methods
like AIPW \eqref{eq:DR} require different nuisance components to evaluate different policies, thus making them
less readily applicable to learning.
That being said, it may still be of interest to compare ADR with AIPW for the task of evaluating a single policy,
and to see whether the form of ADR---optimized for policy learning---sacrifices accuracy when used for evaluation.

To this end, we revisit the two simulation settings discussed above. However, instead of trying to learn the best policy,
we simply seek to evaluate how much the optimal policy improves over a never-treating policy. For ADR and IPW, we
use the same value estimates as were maximized for policy learning. For AIPW, we use a weighted form of \eqref{eq:DR}
as in \citet{thomas2016data}, with value functions estimated by \textit{Q-Eval}, i.e., by a backwards iteration
procedure analogous to \eqref{eq:FQ} that is tailored to evaluating a specific policy as opposed to finding the best policy.
Finally, we also consider \textit{Q-Eval} on its own, by averaging across the learned $Q$ values in the initial state across the initial state distribution and the action distribution that follows the policy of interest.

Overall, as seen in Figure \ref{fig:opt-eval}, the robust methods---ADR and AIPW---substantially outperform both
IPW and \textit{Q-Eval} here, while AIPW is slightly more accurate that ADR. It thus appears that if the only task of
interest is to evaluate a pre-specified policy then AIPW is a good method to start with. However, if there's also
a need to learn policies by empirical maximization, ADR may present a valuable option.

\begin{figure}[!htbp]
	\centering
	\begin{tabular}{cc}
		\includegraphics[width=0.43\columnwidth]{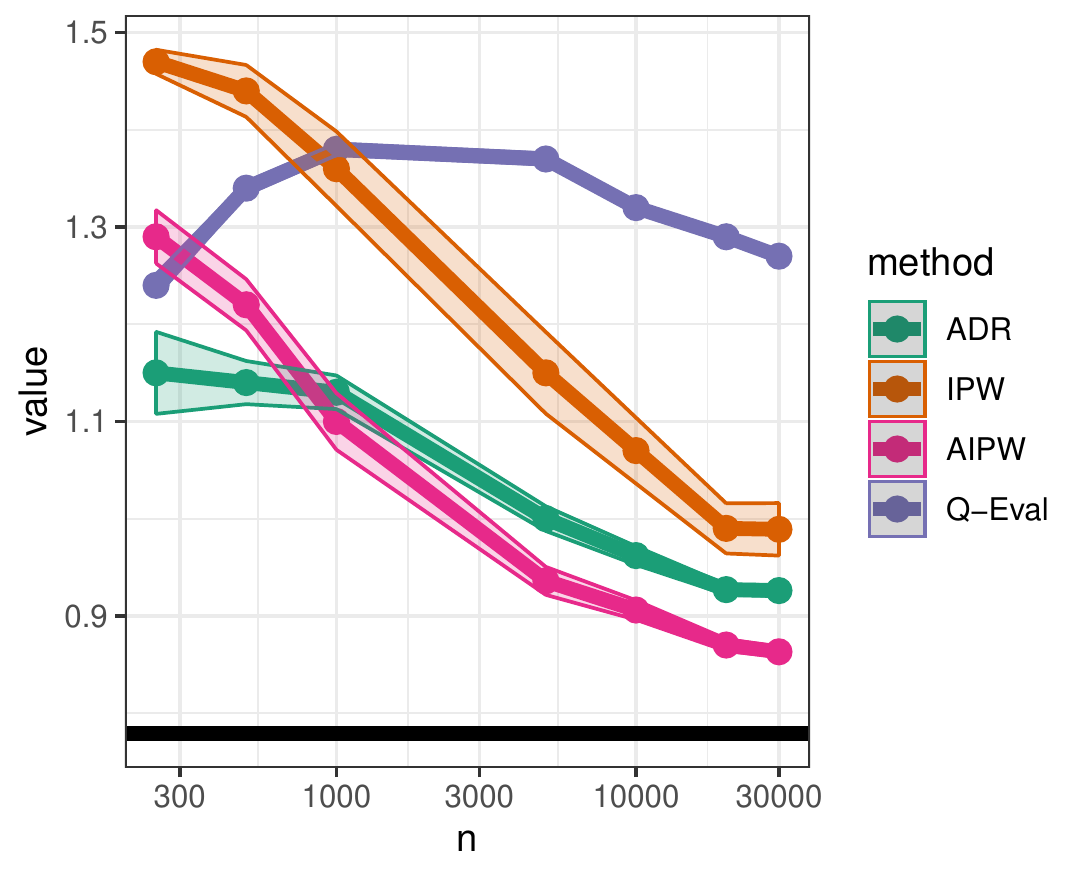} &
		\includegraphics[width=0.43\columnwidth]{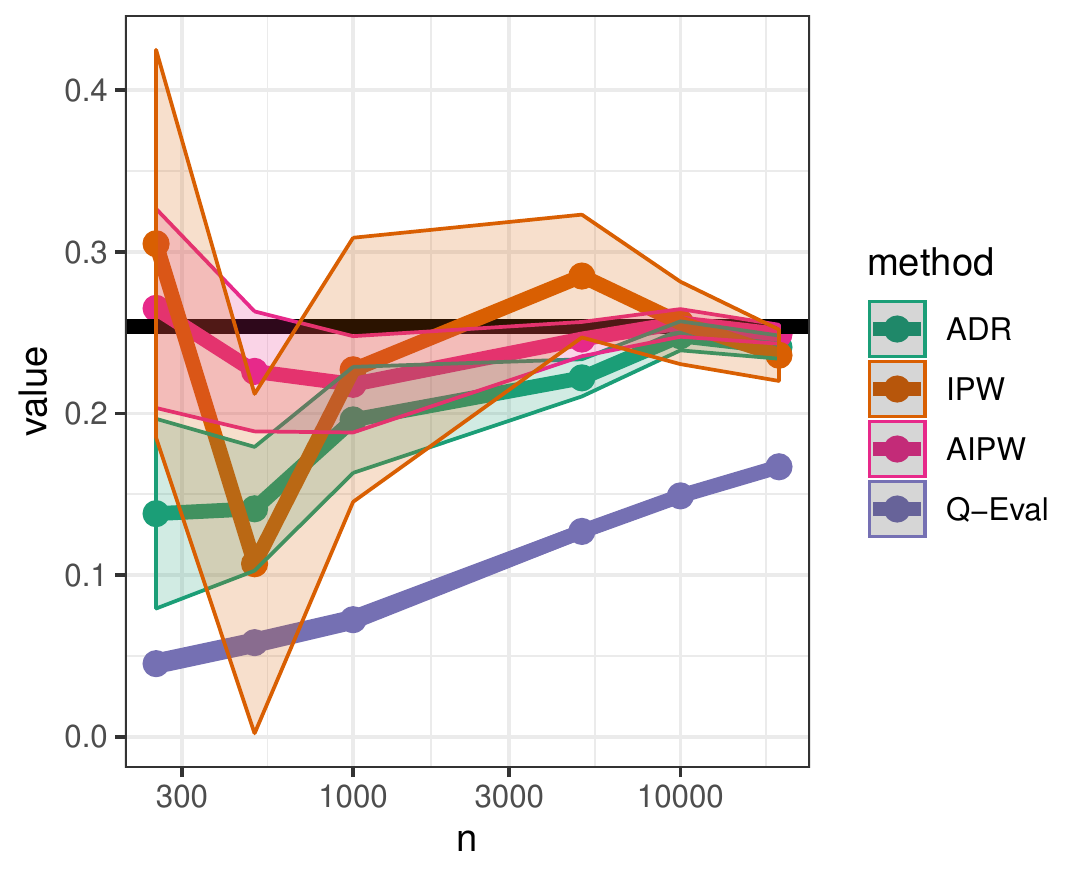}  \\
		\includegraphics[width=0.43\columnwidth]{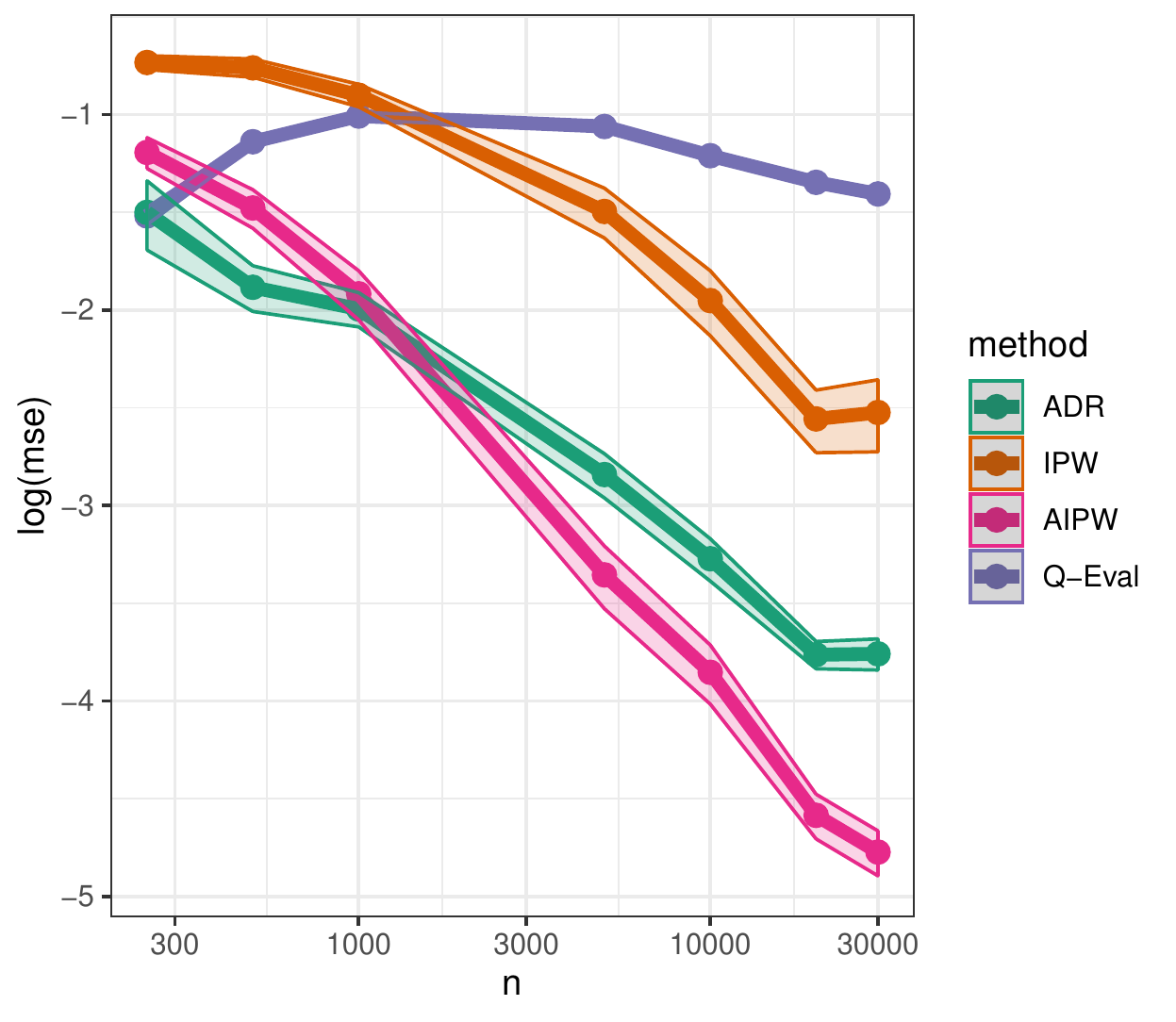} &
		\includegraphics[width=0.43\columnwidth]{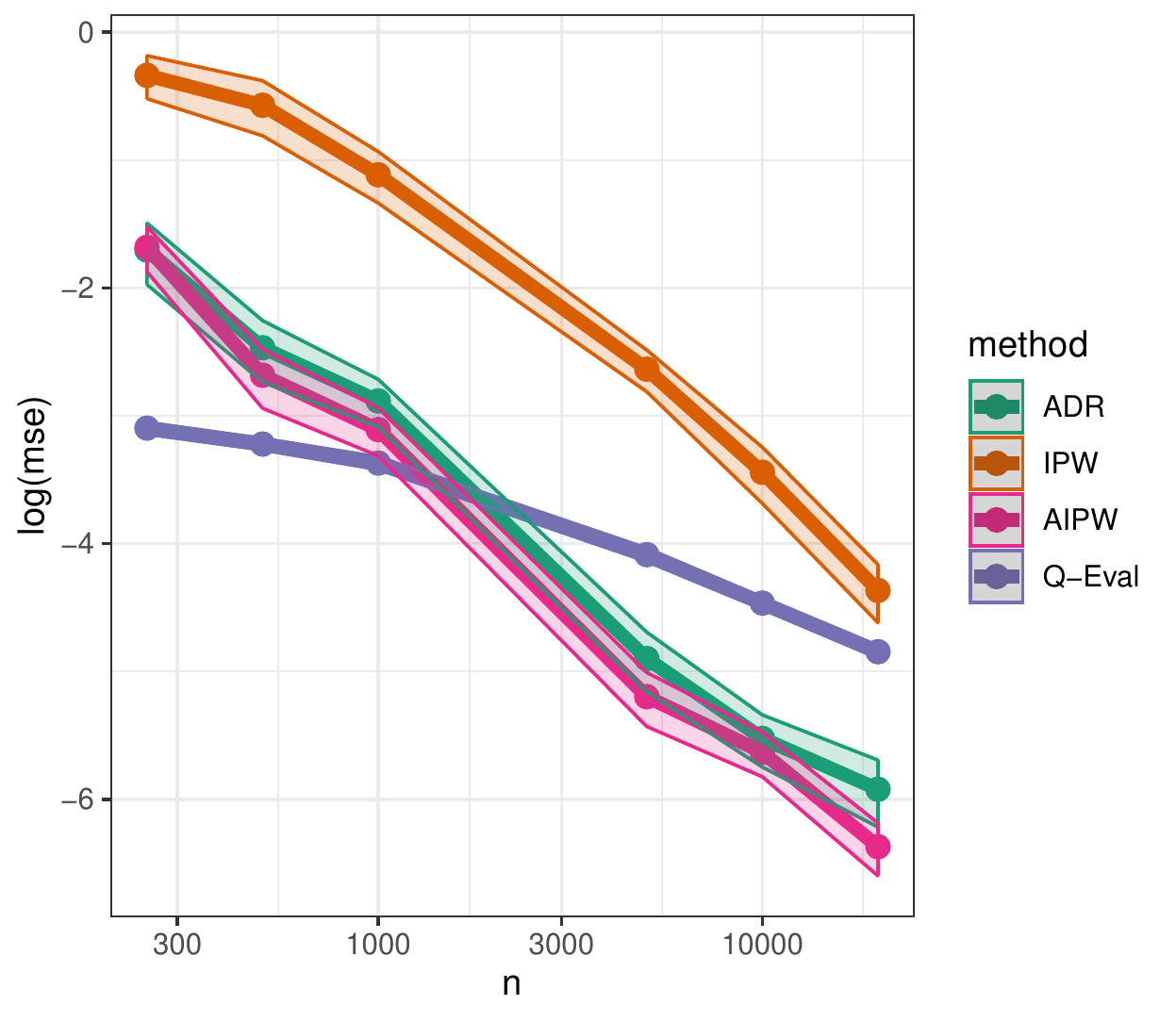}  \\
	\end{tabular}
	\vspace{-0.7\baselineskip}
	\caption{Comparison of ADR, IPW, AIPW, and \textit{Q-Eval} for estimating the value improvement of the best in-class policy
		over the never stop policy.  The left panel is in the setting of Figure \ref{fig:binary} for the binary-treatment setup while the right panel is in the setting of Figure \ref{fig:multiple} for the multiple-treatment setup. The top two figures compare the value estimates of the optimal policy, with the black solid line denotes the true value improvement of the optimal policy via Monte Carlo simulations over 20000 trials. The bottom two figures are the mean squared errors (MSE) of the value estimates on learning the optimal policy. 
		The results here are averaged across 50 independent runs,
		and the shadeded regions denote sampling error. \vspace{-\baselineskip}}
	\label{fig:opt-eval}
\end{figure}

\subsection{Comparison between ADR and Fitted-Q Iteration for Policy Optimization} 

We emphasize that fitted-Q iteration for policy optimization (\textit{Q-Opt}) needs to recursively solve a
series of nonparametric regression problems. \textit{Q-Opt} can be asymptotically biased depending in part on whether the used function approximator can perfectly model the underlying optimal value function. In general, the convergence properties and finite sample performance of \textit{Q-Opt} are not yet well understood and are an active area of research (see e.g.~\cite{szepesvari2005finite,munos2008finite,antos2008learning,chen2019information}). It is thus not unimaginable that, with very little data, \textit{Q-Opt} would regularize towards
a decent model of the world which motivates reasonable decisions; however, once we get more data
and \textit{Q-Opt} increases the complexity of its model fit, the resulting decisions get worse. 

\begin{figure}[!htbp]
	\centering
	\includegraphics[width=\textwidth]{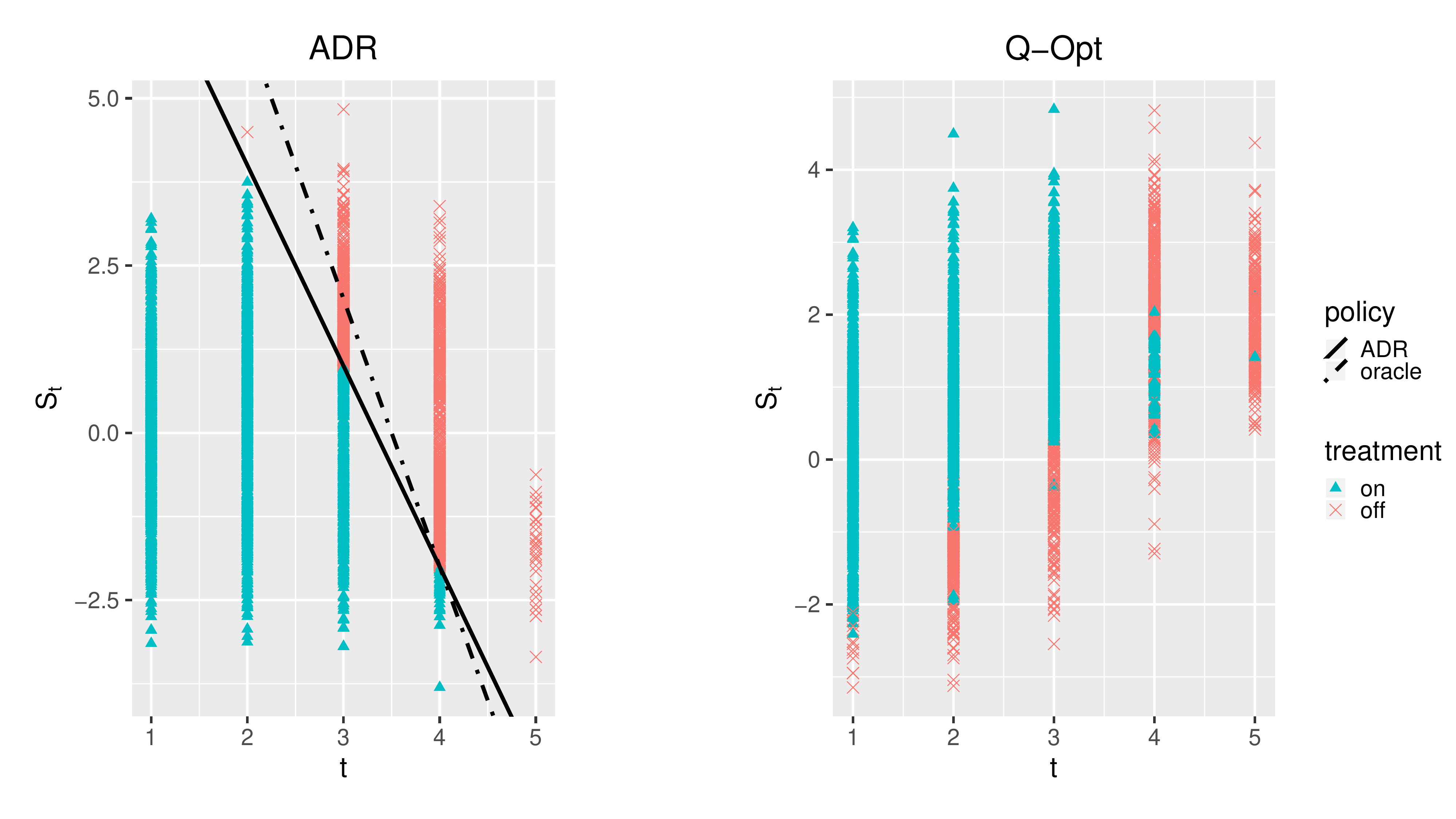}	
	\caption[]
	{{\small A single realization of the best policy learned in the binary action setup case as described in Section \ref{sec:binary}, in the setting of $\nu=0.5, \beta=5, \sigma=1$. ADR  and the oracle choose the best in-class policy from the predefined linear policy class, whereas \textit{Q-Opt} learns the value function via blackbox regression methods and learns a policy that is not so easy to interpret. At each time step, we plot the value of the state $S_t$ from trajectories that have not stopped treating yet.}}    
	\label{fig:compare}
\end{figure}

Finally, as discussed previously, \textit{Q-Opt} seeks to learn the optimal nonparametric policy, whereas ADR aims for the best in-class policy. As argued by, e.g., \citet{athey2017efficient} and \citet{kitagawa2015should},
learning policies that belong to a structured class specified in advance is important in practice, as this
allows stakeholders to enforce constraints such as interpretability, implementability and resource use.
To visualize this point, recall that, in Section \ref{sec:binary}, we used ADR to learn over linear thresholding
policies. In Figure \ref{fig:compare}, we plot the policies learned with different methods for one realization of the simulation in the binary simulation setup in Section \ref{sec:binary}.
The left panel of Figure \ref{fig:compare} shows, at each time step, the value of the state $S_t$ for any trajectory that has not stopped treatment yet according to the policy learned by ADR (shown as a black line). The color coding specifies the policy decision for each trajectory at the given time step.\footnote{There are fewer trajectories plotted as we move along the time axis, because once a trajectory has stopped treatment, it would always stop treatment and there will be no longer decisions made.}
For comparison, we also plot the best policy learned by the oracle using the dotted-dashed line.
The right panel of Figure \ref{fig:compare} is generated the same way, but shows decisions made
by the policy learned using \textit{Q-Opt} instead. 
Unlike ADR, which returns a linear policy, \textit{Q-Opt} learns a policy with a complicated
functional form that is not so easy to interpret.

One might ask whether we could make
\textit{Q-Opt} interpretable by using linear regression in the recursive step \eqref{eq:FQ}. Doing so,
however, would void any nonparametric consistency guaranteed for \textit{Q-Opt}, and in particular
would not recover best-in-class linear policies. The problem is that \textit{Q-Opt} conflates modeling
and policy optimization, rather than separating out these two steps like ADR; in contrast,
we first model $\mu_{now}$ and $\mu_{next}$ using appropriately flexible method and then
choose policy $\hpi$ in a separate optimization step where we can enforce structure.

\section*{Acknowledgement}

We are grateful for enlightening conversations with
Susan Athey,
Miguel Hernan,
Michael Kosorok,
Percy Liang,
Susan Murphy,
Jamie Robins,
Andrea Rotnitzky and
Zhengyuan Zhou,
as well as for helpful comments and feedback from the Associate Editor and
seminar participants at several universities and workshops. 
XN acknowledges the partial support from
the Stanford Data Science Scholars program.
EB acknowledges the partial support of a NSF Career Award and
a Siemens grant.
SW acknowledges the partial support of a Stanford Institute for Human-Centered Artificial Intelligence grant, a Facebook Faculty Award, and National Science Foundation grant DMS-1916163.

\bibliographystyle{plainnat-abbrev}
\bibliography{references}

\newpage
\begin{center}
	{\large\bf SUPPLEMENTARY MATERIAL}
\end{center}
\begin{appendix}
	
	\section{Proofs}

	\begin{proof}[Proof of Lemma \ref{lemm:is-decomp}]
		For a fixed $t$ such that $1\leq t\leq T$,
		\begin{align*}
		&\EE {\mathbbm{1}_{t \geq \tau_\pi}\frac{\mathbbm{1}_{A_{1:t-1}=0}}{\prod_{t'=1}^{t-1}e_{t',0}(S_{1:{t'}})} \p{\mu_{now, W_\pi}(S_{1:t}, t) - \mu_{next, W_\pi}(S_{1:t}, t)}} \\
		&= \int \mathbbm{1}_{t \geq \tau_\pi}\frac{\mathbbm{1}_{A_{1:t-1}=0}}{\prod_{t'=1}^{t-1}e_{t',0}(S_{1:{t'}})} \p{\mu_{now, W_\pi}(S_{1:t}, t) - \mu_{next, W_\pi}(S_{1:t}, t)} f(s_1) \, 
		\\& \ \ \ \ \ \ \ \ \ \prod_{t=2}^T f_t\p{s_t \cond s_{1:(t-1)}, a_{1:(t-1)}} 
		\prod_{t'=1}^{t-1}e_{t',0}(S_{1:{t'}}) dS_{1:t} dA_{1:t-1}\\
		&= \int \mathbbm{1}_{t \geq \tau_\pi} \p{\mu_{now, W_\pi}(S_{1:t}, t) - \mu_{next, W_\pi}(S_{1:t}, t)} f(s_1) \, \prod_{t=2}^T f_t\p{s_t \cond s_{1:(t-1)}, a_{1:(t-1)}} \\
		& \ \ \ \ \ \ \ \ \ \mathbbm{1}_{A_{1:t-1}=0}dS_{1:t} dA_{1:t-1}\\
		&= \EE[\bm{0}] { \mathbbm{1}_{t \geq \tau_\pi} \p{\mu_{now, W_\pi}(S_{1:t}, t) - \mu_{next, W_\pi}(S_{1:t}, t)}}.
		\end{align*}
	\end{proof}
	
	\begin{proof}[Proof of Lemma \ref{lemm:clt}]
		First, we check that $\EE{\tilde{\Delta}(\pi, \bm{0})} = \Delta(\pi, \bm{0})$. It is sufficient to check that for each $t$ such that $1\leq t \leq T$, the following holds.
		\begin{align*}
		&\EE{\mathbbm{1}_{t \geq \tau_\pi} \frac{\mathbbm{1}_{A_{1:t-1} = 0}\mathbbm{1}_{A_t=k} }{\prod_{t'=1}^{t-1}e_{t',0}(S_{1:{t'}})e_{t,k}(S_{1:t}^{(i)})}  \p{Y - \mu_{now, k}(S_{1:t},t)}} \\
		&=\EE{\EE{\mathbbm{1}_{t \geq \tau_\pi} \frac{\mathbbm{1}_{A_{1:t-1} = 0}\mathbbm{1}_{A_t=k} }{\prod_{t'=1}^{t-1}e_{t',0}(S_{1:{t'}})e_{t,k}(S_{1:t})}  \p{Y - \mu_{now, k}(S_{1:t},t)}\cond S_{1:t}}} \\
		&=\EE{\mathbbm{1}_{t \geq \tau_\pi} \EE{ \p{Y - \mu_{now, k}(S_{1:t},t)}\cond S_{1:t}, A_{1:t-1}=0, A_t = k}} \\
		&=0,
		\end{align*}
		and 
		\begin{align*}
		&\EE{\mathbbm{1}_{t \geq \tau_\pi} \frac{\mathbbm{1}_{A_{1:t} = 0}\mathbbm{1}_{A_{t+1}=k} }{\prod_{t'=1}^{t}e_{t',0}(S_{1:{t'}})e_{t+1,k}(S_{1:t+1}^{(i)})}  \p{Y - \mu_{next, k}(S_{1:t},t)}} \\
		&\EE{\EE{\mathbbm{1}_{t \geq \tau_\pi} \frac{\mathbbm{1}_{A_{1:t} = 0}\mathbbm{1}_{A_{t+1}=k} }{\prod_{t'=1}^{t}e_{t',0}(S_{1:{t'}})e_{t+1,k}(S_{1:t+1}^{(i)})}  \p{Y - \mu_{next, k}(S_{1:t},t)}\cond S_{1:t}}} \\
		&\EE{\mathbbm{1}_{t \geq \tau_\pi}\EE{ \frac{\mathbbm{1}_{A_{t+1}=k} }{e_{t+1,k}(S_{1:t+1}^{(i)})}  \p{Y - \mu_{next, k}(S_{1:t},t)}\cond S_{1:t}, A_{1:t}=0}} \\
		&=0,
		\end{align*}
		where the last equality follows from \eqref{eq:rewritenext}.
		Next, we check that the variance $\Omega_\pi <\infty$. 
		Note that, given the assumption that $\abs{Y}$ is bounded and given overlap as in
		Assumption \ref{assu:overlap}, we have
		$\frac{1}{\prod_{t=1}^{T-1} e_{t,0}(S_{1:t})} \frac{1}{e_{t,k}(S_{1:t})} \leq T(1-\eta_0/T)^{-T}/\eta \leq T\exp(2\eta_0)/\eta$
		when $T\geq 2$ and $0<\eta_0<1$. We then conclude that $\Omega_\pi <\infty$. 
		
		The desired result then follows from the Lindeberg–Lévy Central Limit Theorem. 
	\end{proof}

	\begin{proof}[Proof of Lemma \ref{lemm:oracle-regret}]
		Given a policy class $\Pi$, letting $\pi^*=\argmax_{\pi \in \Pi} V_{\pi}$. First, we note that
		\begin{align*}
		R(\tilde{\pi}) 
		&= V_{\pi^*} - V_{\tilde{\pi}} \\
		&= \Delta(\pi^*, \tilde{\pi}) \\
		&= \tilde{\Delta}(\pi^*, \tilde{\pi}) + \Delta(\pi^*, \tilde{\pi}) - \tilde{\Delta}(\pi^*, \tilde{\pi}) \\
		&\leq \Delta(\pi^*, \tilde{\pi}) - \tilde{\Delta}(\pi^*, \tilde{\pi}) \\
		&\leq \sup_{\pi,\pi'\in\Pi}\abs{\Delta(\pi, \pi') - \tilde{\Delta}(\pi, \pi') }.
		\end{align*}
		We note that although we work in the sequential policy learning setup, by redefining policy $\pi$ as in \eqref{eq:pi-new-defn}, the form of $\tilde{\Delta}$ in \eqref{eq:delta-tilde-general} becomes the same as $\tilde{\Delta}$ in \cite{{zhou2018offline}}. Given Assumptions \ref{assu:consistency}--\ref{assu:entropy}, by directly following their Lemma 2, we have for any $\delta >0$, with probability at least $1-2\delta$, there exists universal constants $0<c_1,\,c_2<\infty$ such that
		\begin{align*}
		\sup_{\pi, \pi' \in \Pi} \abs{\tilde{\Delta}(\pi,\pi') - \Delta(\pi,\pi')} \leq \p{c_1 \kappa(\Pi) + c_2 + \sqrt{2\log\frac{1}{\delta}}} \sqrt{\frac{V^*}{n}} + o\p{\frac{1}{\sqrt{n}}}.
		\end{align*}
		Let the last term be $c_n$. Thus for any $c > 0$, there exists $N(c)$ such that for all $n > N(c)$, $c_n < c/\sqrt{n}$. 
		Choose $\varepsilon_0(c,\delta)$ so small such that by letting
		$\p{c_1 \kappa(\Pi) + c_2 + \sqrt{2\log\frac{1}{\delta}}} \sqrt{\frac{V^*}{n}} + c/\sqrt{n} < \varepsilon_0(\delta, c)$, we have $n > N(c)$. The result then follows immediately for all $\varepsilon < \varepsilon_0(\delta, c)$.
	\end{proof}
	
	We now state a simple result on the mean-squared error of noisy products that will be useful
	in the proof of Lemma \ref{lemm:cf-uniform}.
	
	\begin{prop}
		\label{prop:prod_err}
		Let $X_t$ for $t = 1, \, ..., \, T$ be a set of (not necessarily independent or identically distributed)
		random variables with $\abs{\log(1 + X_t)} \leq c$ almost surely and $\EE{\log(1 + X_t)^2} \leq v^2$. Then,
		\begin{equation}
		\label{eq:prod_err}
		\EE{\p{\prod_{t = 1}^T \p{1 + X_t} - 1}^2} \leq v^2 \p{4 + 64 T^2 \p{\frac{e^{2Tc} - 2Tc - 1}{4T^2c^2}}}
		\end{equation}
		for all
		\begin{equation}
		\label{eq:prod_err_cond}
		v^2 \leq \min\cb{1, \, \p{4 T^2 \p{\frac{e^{2Tc} - 2Tc - 1}{4T^2c^2}}}^{-1}}.
		\end{equation}
		\proof
		For convenience, let $a_t = \EE{\log(1 + X_t)}$ and $Z_t = \log(1 + X_t) - a_t$. Then, 
		writing $a =  \sum_{t = 1}^T a_t$, we see that
		\begin{align*}
		&\EE{\p{\prod_{t = 1}^T \p{1 + X_t} - 1}^2}
		= \EE{\p{e^{a + \sum_{t = 1}^T Z_t} - 1}^2} \\
		&\ \ \ \ \ \ \ \ \ \ \leq e^{2a} \EE{e^{2\sum_{t = 1}^T Z_t}} - 2 e^a + 1
		= \p{e^a - 1}^2 + e^{2a} \p{\EE{e^{2\sum_{t = 1}^T Z_t}} - 1},
		\end{align*}
		where the inequality statement follows by Jensen's inequality and convexity of the exponential function.
		By Cauchy-Schwartz, we see that \smash{$\Var{\sum_{t = 1}^T Z_t} \leq T^2 v^2$}, and clearly \smash{$\abs{\sum_{t = 1}^T Z_t} \leq Tc$}.
		Thus, by a Bernstein-style bound on its moment-generating function (see, e.g., Lemma 7.26 of
		\citet{lafferty2008concentration}), we find that
		\begin{align*}
		0 &\leq \EE{e^{2\sum_{t = 1}^T Z_t}} - 1  \leq \exp\sqb{4 T^2 v^2 \p{\frac{e^{2Tc} - 2Tc - 1}{4T^2c^2}}} - 1 \\
		& \leq 8 T^2 v^2 \p{\frac{e^{2Tc} - 2Tc - 1}{4T^2c^2}}
		\end{align*}
		whenever \eqref{eq:prod_err_cond} holds, since $e^x - 1 \leq 2x$ for all $0 \leq x \leq 1$.
		Meanwhile, we must have $\abs{a} \leq v$, and so we conclude that \eqref{eq:prod_err} holds;
		here, we used the bounds $(e^{a} - 1)^2 \leq 4a^2$ and $e^{2a} \leq 8$ for all $\abs{a} \leq 1$.
		\endproof
	\end{prop}

	\begin{proof}[Proof of Lemma \ref{lemm:cf-uniform}]
		For simplicity, we will prove this result in a setting where the nuisance components $\he_{t, k}(s_{1:t})$,
		$\hmu_{now,k}(s_{1:t})$ and $\hmu_{next,k}(s_{1:t})$ are trained on a separate development set and can
		thus be considered as exogenous. Results in our setting of interest, i.e., with cross-fitting, can then be
		derived analogously by applying the same argument multiple times, with one fold as the focal fold and
		the other folds acting as the development set. See \citet{chernozhukov2016double} for further details,
		and the proof of Lemma 4 of \citet{athey2017efficient} for a concrete example of a proof following this strategy.
		
		Let $\delta > 0$ be pre-specified. By Assumption \ref{assu:nuisance}, there is an $n(\delta)$ for which 
		\begin{equation}
		\label{eq:sup_conc}
		\sup\cb{\abs{\he_{t, k}(s_{1:t}) - e_{t,k}(s_{1:t})} : t \in 1, \, \ldots, \, T, \, k \in 1, \, \ldots, \ K, \, s \in \set^T} \leq \frac{1 - \eta_0}{2T}
		\end{equation}
		for all $n \geq n(\delta)$, with probability at least $1 - \delta$. For the rest of this proof, we focus the event under
		which \eqref{eq:sup_conc} has occurred.
		Now, recalling notation from \eqref{eq:plugin-oracle-estimator}, etc., we have
		
		\begin{align*}
		&\abs{\tilde{\Delta}(\pi, \pi')  - \hat{\Delta}(\pi, \pi')} \\
		&\ \ \ \  = \abs{\frac{1}{n} \sum_{i = 1}^n  \sum_{t = 1}^T u_{t,i}^{\pi,\pi'} \p{\frac{\mathbbm{1}_{A_{1:t-1}^{(i)} = 0}}{\prod_{t'=1}^{t-1}e_{t',0}(S_{1:{t'}}^{(i)})}\tilde{\Psi}_{t,W_\pi}(S_{1:t}^{(i)})
				- \frac{\mathbbm{1}_{A_{1:t-1}^{(i)} = 0}}{\prod_{t'=1}^{t-1}\he_{t',0}(S_{1:{t'}}^{(i)})}\hat{\Psi}_{t,W_\pi}(S_{1:t}^{(i)})}} \\
		&\ \ \ \ \leq \sum_{t = 1}^T \underbrace{\abs{\frac{1}{n} \sum_{i = 1}^n
				u_{t,i}^{\pi,\pi'}\p{\frac{\mathbbm{1}_{A_{1:t-1}^{(i)} = 0}}{\prod_{t'=1}^{t-1}e_{t',0}(S_{1:{t'}}^{(i)})}\tilde{\Psi}_{t,W_\pi}(S_{1:t}^{(i)})
					- \frac{\mathbbm{1}_{A_{1:t-1}^{(i)} = 0}}{\prod_{t'=1}^{t-1}\he_{t',0}(S_{1:{t'}}^{(i)})}\hat{\Psi}_{t,W_\pi}(S_{1:t}^{(i)})}}}_{\text{err}_{t}} \\
		& \ \ \ \ \leq T \sup\cb{\text{err}_{t} : t = 1, \, \ldots, \, T},
		\end{align*}
		where $u_{t,i}^{\pi,\pi'} = \mathbbm{1}_{t \geq \tau_\pi} - \mathbbm{1}_{t \geq \tau_{\pi'}}$.
		Furthermore, we note that
		\begin{align*}
		\text{err}_{t}
		&\leq \abs{\frac{1}{n} \sum_{i = 1}^n u_{t,i}^{\pi,\pi'} \p{\frac{\mathbbm{1}_{A_{1:t-1}^{(i)} = 0}}{\prod_{t'=1}^{t-1}e_{t',0}(S_{1:{t'}}^{(i)})}
				- \frac{\mathbbm{1}_{A_{1:t-1}^{(i)} = 0}}{\prod_{t'=1}^{t-1}\he_{t',0}(S_{1:{t'}}^{(i)})}}\tilde{\Psi}_{t,W_\pi}(S_{1:t}^{(i)})} \\
		&\ \ \ \ \ \ + \abs{\frac{1}{n} \sum_{i = 1}^n u_{t,i}^{\pi,\pi'} \frac{\mathbbm{1}_{A_{1:t-1}^{(i)} = 0}}{\prod_{t'=1}^{t-1}\he_{t',0}(S_{1:{t'}}^{(i)})}\p{\tilde{\Psi}_{t,W_\pi}(S_{1:t}^{(i)})
				- \hat{\Psi}_{t,W_\pi}(S_{1:t}^{(i)})}}.
		\end{align*}
		We now proceed to bound the two above summands separately; call them
		$\text{err}_{t}(1)$ and $\text{err}_{t}(2)$. The first term
		measures our errors in estimating the propensity of starting treatment, while the second measures errors in the
		advantage doubly robust scores. For the first term, we further note that
		\begin{align*}
		&\text{err}_{t}(1)
		\leq \Bigg|\frac{1}{n} \sum_{i = 1}^n  u_{t,i}^{\pi,\pi'} \p{\frac{\mathbbm{1}_{A_{1:t-1}^{(i)} = 0}}{\prod_{t'=1}^{t-1}e_{t',0}(S_{1:{t'}}^{(i)})}
			- \frac{\mathbbm{1}_{A_{1:t-1}^{(i)} = 0}}{\prod_{t'=1}^{t-1}\he_{t',0}(S_{1:{t'}}^{(i)})}} \\
		&\ \ \ \ \ \ \ \ \ \ \ \ \ \ \ \ \ \ \times \p{\mu_{now,W_\pi}(S_{1:t}^{(i)},t) - \mu_{next,W_\pi}(S_{1:t}^{(i)}, t)}\Bigg| \\
		&\ \ \ \ \ \ + \Bigg|\frac{1}{n} \sum_{i = 1}^n  u_{t,i}^{\pi,\pi'} \p{\frac{\mathbbm{1}_{A_{1:t-1}^{(i)} = 0}}{\prod_{t'=1}^{t-1}e_{t',0}(S_{1:{t'}}^{(i)})}
			- \frac{\mathbbm{1}_{A_{1:t-1}^{(i)} = 0}}{\prod_{t'=1}^{t-1}\he_{t',0}(S_{1:{t'}}^{(i)})}} \\
		&\ \ \ \ \ \ \ \ \ \ \ \ \ \ \ \ \ \ \times \p{\mathbbm{1}_{A_t^{(i)}=W_\pi} \frac{Y^{(i)} - \mu_{now, W_\pi}(S_{1:t}^{(i)},t)}{e_{t,W_\pi}(S_{1:t}^{(i)})}}\Bigg| \\
		&\ \ \ \ \ \ + \Bigg|\frac{1}{n} \sum_{i = 1}^n  u_{t,i}^{\pi,\pi'} \p{\frac{\mathbbm{1}_{A_{1:t-1}^{(i)} = 0}}{\prod_{t'=1}^{t-1}e_{t',0}(S_{1:{t'}}^{(i)})}
			- \frac{\mathbbm{1}_{A_{1:t-1}^{(i)} = 0}}{\prod_{t'=1}^{t-1}\he_{t',0}(S_{1:{t'}}^{(i)})}} \\
		&\ \ \ \ \ \ \ \ \ \ \ \ \ \ \ \ \ \ \times \p{\mathbbm{1}_{A_{t}^{(i)}=0} \mathbbm{1}_{A_{t+1}^{(i)}=W_\pi} \frac{Y^{(i)} - \mu_{now, W_\pi}(S_{1:t+1}^{(i)},t+1)}{e_{t,0}(S_{1:t}^{(i)}) e_{t+1,W_\pi}(S_{1:t+1}^{(i)})}}\Bigg| \\
		&\ \ \ \ \ \ + \Bigg|\frac{1}{n} \sum_{i = 1}^n  u_{t,i}^{\pi,\pi'} \p{\frac{\mathbbm{1}_{A_{1:t-1}^{(i)} = 0}}{\prod_{t'=1}^{t-1}e_{t',0}(S_{1:{t'}}^{(i)})}
			- \frac{\mathbbm{1}_{A_{1:t-1}^{(i)} = 0}}{\prod_{t'=1}^{t-1}\he_{t',0}(S_{1:{t'}}^{(i)})}}\\
		&\ \ \ \ \ \ \ \ \ \ \ \ \ \ \ \ \ \ \times \p{\mathbbm{1}_{A_{t}^{(i)}=0} \mathbbm{1}_{A_{t+1}^{(i)}=W_\pi} \frac{\mu_{now, W_\pi}(S_{1:t+1}^{(i)},t+1) - \mu_{next, W_\pi}(S_{1:t}^{(i)},t)}{e_{t,0}(S_{1:t}^{(i)}) e_{t+1,W_\pi}(S_{1:t+1}^{(i)})}}\Bigg|.
		\end{align*}
		The first summand above can be bounded by Cauchy-Schwartz:
		\begin{align*}
		&\ldots \leq 
		\sqrt{\frac{1}{n} \sum_{i = 1}^n  \p{\frac{\mathbbm{1}_{A_{1:t-1}^{(i)} = 0}}{\prod_{t'=1}^{t-1}e_{t',0}(S_{1:{t'}}^{(i)})}
				- \frac{\mathbbm{1}_{A_{1:t-1}^{(i)} = 0}}{\prod_{t'=1}^{t-1}\he_{t',0}(S_{1:{t'}}^{(i)})}}^2} \\
		&\ \ \ \ \ \ \ \ \ \ \ \ \ \ \ \ \ \times \sqrt{\frac{1}{n} \sum_{i = 1}^n \p{\mu_{now,W_\pi}(S_{1:t}^{(i)},t) - \mu_{next,W_\pi}(S_{1:t}^{(i)}, t)}^2}.
		\end{align*}
		The second moments of $\mu_{now,W_\pi}(S_{1:t}^{(i)},t) - \mu_{next,W_\pi}(S_{1:t}^{(i)}, t)$ are governed by Assumption \ref{assu:delta},
		so we know that the second term in the product is bounded on the order of $n^{-\kappa_\delta}$ (with constants that do not depend
		on the problem setting. Meanwhile, 
		\begin{align*}
		&\frac{1}{n} \sum_{i = 1}^n  \p{\frac{\mathbbm{1}_{A_{1:t-1}^{(i)} = 0}}{\prod_{t'=1}^{t-1}e_{t',0}(S_{1:{t'}}^{(i)})}
			- \frac{\mathbbm{1}_{A_{1:t-1}^{(i)} = 0}}{\prod_{t'=1}^{t-1}\he_{t',0}(S_{1:{t'}}^{(i)})}}^2 \\
		&\ \ \ \ \ \ \ \ \ \ \ \ \leq \frac{1}{n} \sum_{i = 1}^n \p{\frac{\mathbbm{1}_{A_{1:t-1}^{(i)} = 0}}{\prod_{t'=1}^{t-1}e_{t',0}(S_{1:{t'}}^{(i)})}}^2
		\p{1 - \frac{\prod_{t'=1}^{t-1}e_{t',0}(S_{1:{t'}}^{(i)})}{\prod_{t'=1}^{t-1}\he_{t',0}(S_{1:{t'}}^{(i)})}}^2 \\
		&\ \ \ \ \ \ \ \ \ \ \ \ \leq \frac{e^{2\eta_0}}{n} \sum_{i = 1}^n  \p{1 - \frac{\prod_{t'=1}^{t-1}e_{t',0}(S_{1:{t'}}^{(i)})}{\prod_{t'=1}^{t-1}\he_{t',0}(S_{1:{t'}}^{(i)})}}^2,
		\end{align*}
		where on the last line we used Assumption \ref{assu:overlap}. Finally, to bound this last term, notice that from \eqref{eq:sup_conc} and Assumption \ref{assu:overlap}, we have $\inf_{t,k} \hat{e}_{t,k}(\cdot) \geq 1-\frac{\eta_0+1}{2T}$ almost surely.  We now aim to bound this last term with Proposition \ref{prop:prod_err}. By taking $X_t := \frac{e_{t,0}(S_{1:t})}{\hat{e}_{t,0}(S_{1:t})} - 1$, we have
		\begin{align*}
		\abs{\log(X_t+1)} 
		&= \abs{\log\p{1+\frac{e_{t,0}(S_{1:t})-\hat{e}_{t,0}(S_{1:t})}{\hat{e}_{t,0}(S_{1:t})}}}\\
		&\leq \max\{\log(1+\frac{1-\eta_0}{2T-\eta_0-1}), -\log(1-\frac{1-\eta_0}{2T-\eta_0-1})\}\\
		&\leq \frac{2-2\eta_0}{2T-\eta_0-1}
		\end{align*}
		almost surely for $T\geq 2$, where we use the fact that $\log(1+x) < 2x$ for $x>0$ and $\log(1-x) > -2x$ for $0<x<0.5$.
		
		We also have 
		\begin{align*}
		\EE{(\log(X_t+1))^2} 
		&= \EE{(\log e_{t,0}(S_{1:t}) - \log \hat{e}_{t,0}(S_{1:t}) )^2} \\
		&\leq \EE{1/\p{\min\{1-\eta_0/T,  1-(\eta_0+1)/(2T)\}}^2 (e_{t,0}(S_{1:t})-\hat{e}_{t,0}(S_{1:t}))^2}\\
		&\leq \frac{(1-\eta_0/T)^2 (1-(\eta_0+1)/(2T))^2}{\p{\min\{1-\eta_0/T,  1-(\eta_0+1)/(2T)\}}^2} C_e n^{-2\kappa_e} \\
		&\leq C_e n^{-2\kappa_e},
		\end{align*}
		where we have used the fact that $\abs{\log a - \log b} \leq \frac{\abs{a-b}}{\min\{a,b\}}$ for $a,b>0$ by concavity, and Assumption \ref{assu:overlap} and \ref{assu:nuisance}, and assuming \eqref{eq:sup_conc} holds.
		
		Given the above, applying Proposition \eqref{prop:prod_err}, we see that for $n > n(\eta,T)$ where \eqref{eq:prod_err_cond} would hold for $n(\eta,T)$, we  have for some constant $C$,
		\begin{align}
		\EE{\p{1 - \frac{\prod_{t'=1}^{t-1}e_{t',0}(S_{1:{t'}}^{(i)})}{\prod_{t'=1}^{t-1}\he_{t',0}(S_{1:{t'}}^{(i)})}}^2} 
		\leq C n^{-2\kappa_e} T^2.
		\end{align}
		Thus, we conclude that the first summand in the decomposition of $\text{err}_{t}(1)$ is bounded on the order of $T n^{-(\kappa_e + \kappa_\delta)}$ in probability whenever \eqref{eq:sup_conc} holds.
		
		Meanwhile, the second and third summands used to construct $\text{err}_{t}(1)$ multiply errors in $\he_{t',0}(S_{1:t'})$ with
		noise terms that are mean-zero conditionally on $\ff_t$ and $\ff_{t+1}$ respectively thanks to Assumption \ref{assu:ignorability}
		(recall that we assume that $\he_{t',0}(S_{1:t'})$ has been trained
		on a separate development set and so the predictive surface $\he_{t',0}(\cdot)$ may be taken as exogenous).
		They can thus be bounded to order $n^{-(1/2 + \kappa_e)}$ whenever
		$\he_{t',0}(S_{1:t'})$ is consistent in squared-error loss at rate $n^{-\kappa_e}$; see, e.g., the proof of Lemma 4 of \citet{athey2017efficient}
		for details.
		Finally, the fourth term can again be bounded using Cauchy-Schwartz analogously to the first one, noting that
		$\mu_{now, W_\pi}(S_{1:t+1}^{(i)},t+1) - \mu_{next, W_\pi}(S_{1:t}^{(i)},t)$ is bounded on the order of $n^{-\kappa_\gamma}$
		by Assumption \ref{assu:delta}.
		
		Thus, we find that there exists a constant $C_1(\delta)$ that depends the constants in
		Assumptions \ref{assu:nuisance} and \ref{assu:delta}, as well as a problem-specific constant $n_1(\delta, \, \eta, \, T)$ such that
		$$  \text{err}_{t} \leq C_1(\delta) T n^{-\p{\kappa_e + \min\cb{1/2, \ \kappa_\delta, \, \kappa_\gamma}}} $$
		with probability at least $1 - \delta$, for all $n \geq n_1(\delta, \, \eta, \, T)$ and on the event \eqref{eq:sup_conc}.
		
		We now move to bounding $\text{err}_{t}(2)$. To this end, we first note that
		$$ \hG_{t,i}^{\pi,\pi'} := u_{t,i}^{\pi,\pi'} \frac{\mathbbm{1}_{A_{1:t-1}^{(i)} = 0}}{\prod_{t'=1}^{t-1}\he_{t',0}(S_{1:{t'}}^{(i)})} $$
		is $\ff_t$-measurable (again, recall that all nuisance components are taken to be exogenous), and
		that it is uniformly bounded on the event \eqref{eq:sup_conc},
		\begin{equation}
		\label{eq:Gbound}
		\abs{\hG_{t,i}^{\pi,\pi'}} \leq e^{1 + \eta_0},
		\end{equation}
		by Assumption \ref{assu:overlap}. Given these preliminaries, we can bound
		$\text{err}_{t}(2)$ by expanding out the square as follows:
		\begin{align*}
		&\text{err}_{t}(2) \leq \abs{\frac{1}{n} \sum_{i = 1}^n \hG_{t,i}^{\pi,\pi'} \p{ \mu_{now,W_\pi}(S_{1:t}^{(i)},t) - \hmu_{now,W_\pi}(S_{1:t}^{(i)},t)  } \p{1 - \frac{\mathbbm{1}_{A_t^{(i)}=W_\pi}}{e_{t,W_\pi}(S_{1:t}^{(i)})}}} \\
		&\ \ \ \ +  \abs{\frac{1}{n} \sum_{i = 1}^n \hG_{t,i}^{\pi,\pi'} \mathbbm{1}_{A_t^{(i)}=W_\pi} \p{ Y^{(i)} - \mu_{now, W_\pi}(S_{1:t}^{(i)},t) } \p{\frac{1}{\he_{t,W_\pi}(S_{1:t}^{(i)})} - \frac{1}{e_{t,W_\pi}(S_{1:t}^{(i)})}}} \\
		&\ \ \ \ +  \Bigg|\frac{1}{n} \sum_{i = 1}^n \hG_{t,i}^{\pi,\pi'} \mathbbm{1}_{A_t^{(i)}=W_\pi}  \p{ \hmu_{now,W_\pi}(S_{1:t}^{(i)},t) - \mu_{now,W_\pi}(S_{1:t}^{(i)},t)}\\
		&\ \ \ \ \ \ \ \ \ \ \ \ \ \ \ \ \ \ \times\p{\frac{1}{\he_{t,W_\pi}(S_{1:t}^{(i)})} - \frac{1}{e_{t,W_\pi}(S_{1:t}^{(i)})}}  \Bigg| \\
		&\ \ \ \ + \abs{\frac{1}{n} \sum_{i = 1}^n \hG_{t,i}^{\pi,\pi'} \p{ \mu_{next,W_\pi}(S_{1:t}^{(i)},t) - \hmu_{next,W_\pi}(S_{1:t}^{(i)},t)  } \p{1 - \frac{\mathbbm{1}_{A_t^{(i)}=0,A_{t+1}^{(i)}=W_\pi}}{e_{t,0}(S_{1:t}^{(i)})e_{t+1,W_\pi}(S_{1:t+1}^{(i)})}}} \\
		&\ \ \ \ +  \Bigg|\frac{1}{n} \sum_{i = 1}^n \hG_{t,i}^{\pi,\pi'} \mathbbm{1}_{A_t^{(i)}=0,A_{t+1}^{(i)}=W_\pi} \p{ Y^{(i)} - \mu_{next, W_\pi}(S_{1:t}^{(i)},t) } \\
		&\ \ \ \ \ \ \ \ \ \ \ \ \ \ \ \ \ \ \ \ \times \p{\frac{1}{\he_{t,0}(S_{1:t}^{(i)})\he_{t+1,W_\pi}(S_{1:t+1}^{(i)})} - \frac{1}{e_{t,0}(S_{1:t}^{(i)})e_{t+1,W_\pi}(S_{1:t+1}^{(i)})}}\Bigg| \\
		&\ \ \ \ +  \Bigg|\frac{1}{n} \sum_{i = 1}^n \hG_{t,i}^{\pi,\pi'} \mathbbm{1}_{A_t^{(i)}=0,A_{t+1}^{(i)}=W_\pi}  \p{ \hmu_{next,W_\pi}(S_{1:t}^{(i)},t) - \mu_{next,W_\pi}(S_{1:t}^{(i)},t)} \\
		&\ \ \ \ \ \ \ \ \ \ \ \ \ \ \ \ \ \ \ \ \times \p{\frac{1}{\he_{t,0}(S_{1:t}^{(i)})\he_{t+1,W_\pi}(S_{1:t+1}^{(i)})} - \frac{1}{e_{t,0}(S_{1:t}^{(i)})e_{t+1,W_\pi}(S_{1:t+1}^{(i)})}}\Bigg|. 
		\end{align*}
		Here, using Assumption \ref{assu:ignorability}, we again see that the 1st, 2nd and 4th terms are sums of regression error
		multiplied by conditionally mean-zero noise and so, as argued
		above, decay at rates $n^{-(1/2 + \kappa_e)}$ or $n^{-(1/2 + \kappa_\mu)}$ depending on the rate of convergence of the nuisance
		components. Meanwhile, the 3rd and
		6th terms can be bounded by Cauchy-Schwartz, which imply that they decay on the order of
		$n^{-( \kappa_e + \kappa_\mu)}$.
		The last remaining term to bound is the 5th summand in the above bound for $\text{err}_{t}(2)$, which we can further
		decompose as
		\begin{align*}
		&\ldots \leq
		\Bigg|\frac{1}{n} \sum_{i = 1}^n \hG_{t,i}^{\pi,\pi'} \mathbbm{1}_{A_t^{(i)}=0,A_{t+1}^{(i)}=W_\pi} \p{ Y^{(i)} - \mu_{now, W_\pi}(S_{1:t+1}^{(i)},t+1) } \\
		&\ \ \ \ \ \ \ \ \ \ \ \ \ \ \ \ \ \ \ \ \times \p{\frac{1}{\he_{t,0}(S_{1:t}^{(i)})\he_{t+1,W_\pi}(S_{1:t+1}^{(i)})} - \frac{1}{e_{t,0}(S_{1:t}^{(i)})e_{t+1,W_\pi}(S_{1:t+1}^{(i)})}}\Bigg| \\
		&\ \ \ \ + \Bigg|\frac{1}{n} \sum_{i = 1}^n \hG_{t,i}^{\pi,\pi'} \mathbbm{1}_{A_t^{(i)}=0,A_{t+1}^{(i)}=W_\pi} \p{ \mu_{now, W_\pi}(S_{1:t+1}^{(i)},t+1) - \mu_{next, W_\pi}(S_{1:t}^{(i)},t) } \\
		&\ \ \ \ \ \ \ \ \ \ \ \ \ \ \ \ \ \ \ \ \times \p{\frac{1}{\he_{t,0}(S_{1:t}^{(i)})\he_{t+1,W_\pi}(S_{1:t+1}^{(i)})} - \frac{1}{e_{t,0}(S_{1:t}^{(i)})e_{t+1,W_\pi}(S_{1:t+1}^{(i)})}}\Bigg|.
		\end{align*}
		The first of these two terms can again be bounded via Assumption \ref{assu:ignorability}; meanwhile, the second one can be
		bounded by Cauchy-Schwartz to order $n^{-(\kappa_e + \kappa_\gamma)}$ using Assumptions \ref{assu:nuisance} and
		\ref{assu:delta}.
		
		Tallying up all our bounds, we have found that, there is a constant $C(\delta)$ depending on $\delta$ and the constants used in Assumptions
		\ref{assu:overlap}, \ref{assu:nuisance} and \ref{assu:delta} such that, and a problem-specific threshold $n(\delta, \,\eta, \,T)$ such that
		$$ \abs{\tilde{\Delta}(\pi, \pi')  - \hat{\Delta}(\pi, \pi')} \leq C(\delta) T^2 n^{-\min\cb{1/2 + \kappa_e, \, 1/2 + \kappa_\mu, \, \kappa_e + \kappa_\mu, \, \kappa_e + \kappa_\delta, \, \kappa_e + \kappa_\gamma}}, $$
		with probability at least $1 - 3\delta$, for all $n \geq n(\delta, \,\eta, \,T)$. The desired conclusion follows.
	\end{proof}

	\begin{proof}[Proof of Lemma \ref{lemm:murphy-decomp-term}]
		Following the proof of Lemma \ref{lemm:murphy-decomp}, we obtain an analogue to \eqref{eq:murphy}:
		\begin{equation}
		\label{eq:murphy-term}
		\Delta(\pi, \bm{0}) 
		= -\EE[\bm{0}] {\sum_{t=1}^T \mathbbm{1}_{t \geq \tau_\pi, \, S_t \neq \Phi} \p{Q_{\pi,t}(S_{1:t}, \bm{0}_{1:t}) - \mu_{\pi,t}(S_{1:t}, \bm{0}_{1:(t-1)})}}, 
		\end{equation}
		and again find that $\mu_{now,k}(\cdot)$ can be used to express the terms
		\smash{$\mathbbm{1}_{t \geq \tau_\pi, \, S_t \neq \Phi} \, \mu_{\pi,t}(S_{1:t}, \bm{0}_{1:(t-1)})$}.
		However, given the possibility of terminal states, we now need to account for the possibility that the patient may
		enter the set $\Phi$ at time $S_{t+1}$ when characterizing terms involving the $Q$-function.
		To this end, define policy $\pi_k^{now}$ as follows:  \smash{$\pi^{now}_k(S_{1:t}, A_{t-1}) = 0$} if $S_t \in \Phi$; else, \smash{$\pi^{now}_k(S_{1:t}, A_{t-1}) = A_{t-1}$}
		if $A_{t-1} \neq 0$, and otherwise \smash{$\pi^{now}_k(S_{1:t}, A_{t-1}) = k$}.
		Notice that this policy satisfies Definition \ref{def:regular_policy_term} and that, for all $t \geq \tau_\pi$,
		our policy of interest $\pi$ matches \smash{$\pi^{now}_{W_\pi}$}. Thus, we see that
		\begin{align*}
		&\mathbbm{1}_{t \geq \tau_\pi, \, S_t \neq \Phi}\, Q_{\pi,t}(S_{1:t}, \bm{0}_{1:t})
		= \mathbbm{1}_{t \geq \tau_\pi, \, S_t \neq \Phi}\, Q_{\pi^{now}_{W_\pi},t}(S_{1:t}, \bm{0}_{1:t}) \\
		&\ \ \ \ \ = \mathbbm{1}_{t \geq \tau_\pi, \, S_t \neq \Phi}\, \EE[\pi^{now}_{W_\pi}]{Y \cond S_{1:t}, \, A_{1:t} = \bm{0}_{1:t}}
		=  \mathbbm{1}_{t \geq \tau_\pi, \, S_t \neq \Phi}\, \mu^\Phi_{next, W_\pi}(S_{1:t}, t),
		\end{align*}
		where the last statement follows by analogous derivations to those used in the proof of Lemma \ref{lemm:murphy-decomp} under Assumptions \ref{assu:consistency} and \ref{assu:ignorability}.
	\end{proof}
	
	\section{Simulation Details and Results}
	First, we write out the weighted formula for the ADR estimator with terminal states. We use this form in the experiments for the multiple-action setup where censoring due to death is involved.
	
	\begin{equation*}
	\begin{split}
	&\hat{\Delta}^{\textrm{W}}(\pi, \bm{0})^\Phi \\
	&\ \ =  \sum_{t=1}^T  \frac{\sum_{i=1}^{n}  \frac{\mathbbm{1}_{A_{1:t-1}^{(i)} = 0}}{\prod_{t'=1}^{t-1}\hat{e}^{-q(i)}_{t',0}(S_{1:{t'}}^{(i)})}\mathbbm{1}_{t \geq \tau_\pi^{(i)}}  \mathbbm{1}_{S_{t}^{(i)} \neq \Phi}}{\sum_{i=1}^{n}  \frac{\mathbbm{1}_{A_{1:t-1}^{(i)} = 0}}{\prod_{t'=1}^{t-1}\hat{e}^{-q(i),0}_{t'}(S_{1:{t'}}^{(i)})}} \p{\hat{\mu}^{-q(i)}_{now,W_\pi}(S_{1:t}^{(i)}) -\hat{\mu}^{-q(i)}_{next,W_\pi}(S_{1:t}^{(i)})^\Phi} \\
	&\ \ \ \ +  \sum_{t=1}^T  \frac{\sum_{i=1}^{n}  \frac{\mathbbm{1}_{A_{1:t-1}^{(i)} = 0} \mathbbm{1}_{A_t^{(i)}=W_\pi}}{\prod_{t'=1}^{t-1}\hat{e}^{-q(i)}_{t',0}(S_{1:{t'}}^{(i)})\hat{e}^{-q(i)}_{t,W_\pi}(S_{1:{t}}^{(i)}))} \mathbbm{1}_{t \geq \tau_\pi^{(i)}}  \mathbbm{1}_{S_{t}^{(i)} \neq \Phi}\Bigg(  Y^{(i)} -\hat{\mu}^{-q(i)}_{now,W_\pi}(S_{1:t}^{(i)})\Bigg)}{\sum_{i=1}^{n}  \frac{\mathbbm{1}_{A_{1:t-1}^{(i)} = 0} \mathbbm{1}_{A_t^{(i)}=W_\pi}}{\prod_{t'=1}^{t-1}\hat{e}^{-q(i)}_{t',0}(S_{1:{t'}}^{(i)}) \hat{e}^{-q(i)}_{t,W_\pi}(S_{1:{t}}^{(i)})} \mathbbm{1}_{t \geq \tau_\pi^{(i)}}  \mathbbm{1}_{S_{t}^{(i)} \neq \Phi}+   \frac{\mathbbm{1}_{A_{1:t-1}^{(i)} = 0}}{\prod_{t'=1}^{t-1}\hat{e}^{-q(i)}_{t',0}(S_{1:{t'}}^{(i)})} \p{1-\mathbbm{1}_{t \geq \tau_\pi^{(i)}}  \mathbbm{1}_{S_{t}^{(i)} \neq \Phi}}}\\
	&\ \ \ \ -  \sum_{t=1}^T  \frac{\sum_{i=1}^{n}  \frac{\mathbbm{1}_{A_{1:t}^{(i)} = 0}\mathbbm{1}_{A_{t+1}^{(i)}=W_\pi} \mathbbm{1}_{S_{t+1}^{(i)} \neq \Phi}}{\prod_{t'=1}^{t}\hat{e}^{-q(i)}_{t',0}(S_{1:{t'}}^{(i)}) \hat{e}^{-q(i)}_{t+1,W_\pi}(S_{1:{t+1}}^{(i)})} \mathbbm{1}_{t \geq \tau_\pi^{(i)}}  \mathbbm{1}_{S_{t}^{(i)} \neq \Phi} \p{Y^{(i)}-\hat{U}^{-q(i)}(S_{1:t}^{(i)}, \Phi)}}{\sum_{i=1}^{n}  A + B+ C},
	\end{split}
	\end{equation*}
	where 
	\begin{align*}
	&A=\frac{\mathbbm{1}_{A_{1:t}^{(i)} = 0} \mathbbm{1}_{A_{t+1}^{(i)}=W_\pi}\mathbbm{1}_{S_{t+1}^{(i)} \neq \Phi}}{\prod_{t'=1}^{t}\hat{e}^{-q(i)}_{t',0}(S_{1:{t'}}^{(i)})\hat{e}^{-q(i)}_{t+1,W_\pi}(S_{1:t+1}^{(i)})  } \mathbbm{1}_{t \geq \tau_\pi^{(i)}}  \mathbbm{1}_{S_{t}^{(i)} \neq \Phi}, \\
	&B =\frac{\mathbbm{1}_{A_{1:t}^{(i)} = 0} \mathbbm{1}_{S_{t+1}^{(i)} = \Phi}}{\prod_{t'=1}^{t}\hat{e}^{-q(i)}_{t',0}(S_{1:{t'}}^{(i)})} \mathbbm{1}_{t \geq \tau_\pi^{(i)}}  \mathbbm{1}_{S_{t}^{(i)} \neq \Phi}, \\
	&C=\frac{\mathbbm{1}_{A_{1:t-1}^{(i)} = 0}}{\prod_{t'=1}^{t-1}\hat{e}^{-q(i)}_{t',0}(S_{1:{t'}}^{(i)})}\p{1-\mathbbm{1}_{t \geq \tau_\pi^{(i)}}  \mathbbm{1}_{S_{t}^{(i)} \neq \Phi}}.
	\end{align*}
	For the binary treatment choices setup as described in Section \ref{sec:binary}, we define a linear thresholding rule $\theta_1 S_t \geq \theta_2 t +  \theta_3$ such that whenever this holds, we stop the treatment. We search over three classes of policies:
	\begin{itemize}
		\item Policies that always stop after some time $t$, corresponding to $\theta_1=0, \theta_2 = -1, \theta_3 \in [1,2,\cdots, T+1]$ where $T=10$ is the horizon length of the study.
		\item Policies that always stop once the patient's state $S_t$ is above some threshold, corresponding to $\theta_1 = 1, \theta_2 = 0, \theta_3 \in [-0.5, 0, 0.5, \cdots, 4.5, 5]$.
		\item Policies that depend on both the time and the patient's state $S_t$, corresponding to $\theta_1=1, \theta_2 \in [-1/4, -1/3, -1/2, -1, -2, -3, -4], \theta_3 \in [1,2,\cdots, 15]$. 
	\end{itemize}
	
	For the multiple treatment choices setup as descibed in Section \ref{sec:multiple}, we consider the linear thresholding class: $\theta_1 X_t' + \theta_2 Y_t' + \theta_3 t \geq \theta_4$ is the region in which we start treatment. If in addition, $\theta_5 X_t' + \theta_6 Y_t' + \theta_7 t \geq \theta_8$, we use the invasive treatment and otherwise, use the non-invasive treatment. We search over the following classes of policies:
	\begin{itemize}
		\item Policies that always start treating at sometime, and always assign the non-invasive treatment option, corresponding to $\theta_1=0, \theta_2=0, \theta_3=1, \theta_4\in [1,2,\cdots, T+1]$ where $T=10$ is the horizon length of the study, $\theta_5=0, \theta_6=0, \theta_7=0,\theta_8=1$
		\item Policies that always start treating at sometime, and always assign the  invasive treatment option, corresponding to $\theta_1=0, \theta_2=0, \theta_3=1, \theta_4\in [1,2,\cdots, T+1]$ where $T=10$ is the horizon length of the study, $\theta_5=0, \theta_6=0, \theta_7=1,\theta_8=0$
		\item Policies that depend on both the time and the two covariates in the form of $\theta_1 \in [0.2,0.7, 1, 3,5], \theta_2 = 1, \theta_3 \in [0,1], \theta_4 \in [1,3,5, 7, 9], \theta_5 \in [0.1 ,0.35, 0.5, 1.5,2.5], \theta_6 = 1, \theta_7 \in [0,1], \theta_8 \in [1,3,5, 7, 9]$. 
		\item Policies that depend on both the time and the two covariates in the form of $\theta_1 \in [0.2,0.7, 1, 3,5], \theta_2 = 1, \theta_3 \in [0,1], \theta_4 \in [1,3,5, 7, 9], \theta_5 =-0.5, \theta_6 = 1, \theta_7 \in [0,1], \theta_8 \in [-5, -2, 1, 4]$.
		
	\end{itemize}
	\label{append:plots}
\begin{table}[ht]
\centering
\begin{tabular}{cccccccccc}
  \hline
n & $\nu$ & $\beta$ & $\sigma$ & ADR & IPW & Q-Opt & Oracle & MSE:ADR & MSE:IPW \\ 
  \hline
250 & 0 & 0.5 & 1 & 8.39e-01 & 8.38e-01 & \bf 8.53e-01 & 8.78e-01 & \bf 1.70e-02 & 8.63e-02 \\ 
  250 & 0 & 0.5 & 3 & \bf 9.11e-01 & 8.25e-01 & 8.79e-01 & 9.25e-01 & \bf 2.42e-02 & 6.63e-02 \\ 
  250 & 0.5 & 0.5 & 1 & 8.32e-01 & \bf 8.33e-01 & 8.02e-01 & 8.76e-01 & \bf 3.51e-02 & 8.27e-02 \\ 
  250 & 0.5 & 0.5 & 3 & \bf 9.18e-01 & 8.61e-01 & 8.78e-01 & 9.27e-01 & \bf 2.51e-02 & 6.85e-02 \\ 
  500 & 0 & 0.5 & 1 & \bf 8.73e-01 & 8.49e-01 & 8.71e-01 & 8.78e-01 & \bf 5.50e-03 & 8.11e-02 \\ 
  500 & 0 & 0.5 & 3 & \bf 9.23e-01 & 8.66e-01 & 8.79e-01 & 9.25e-01 & \bf 5.76e-03 & 6.60e-02 \\ 
  500 & 0.5 & 0.5 & 1 & \bf 8.69e-01 & 8.62e-01 & 8.11e-01 & 8.76e-01 & \bf 3.23e-02 & 7.54e-02 \\ 
  500 & 0.5 & 0.5 & 3 & \bf 9.23e-01 & 8.77e-01 & 8.84e-01 & 9.27e-01 & \bf 1.44e-02 & 6.07e-02 \\ 
  1000 & 0 & 0.5 & 1 & \bf 8.78e-01 & 8.71e-01 & 8.74e-01 & 8.78e-01 & \bf 4.50e-03 & 8.06e-02 \\ 
  1000 & 0 & 0.5 & 3 & \bf 9.25e-01 & 8.75e-01 & 8.87e-01 & 9.25e-01 & \bf 6.67e-03 & 5.76e-02 \\ 
  1000 & 0.5 & 0.5 & 1 & \bf 8.73e-01 & 8.70e-01 & 8.16e-01 & 8.76e-01 & \bf 2.25e-02 & 7.19e-02 \\ 
  1000 & 0.5 & 0.5 & 3 & \bf 9.27e-01 & 8.88e-01 & 8.84e-01 & 9.27e-01 & \bf 8.94e-03 & 5.64e-02 \\ 
  5000 & 0 & 0.5 & 1 & \bf 8.81e-01 & 8.79e-01 & 8.79e-01 & 8.78e-01 & \bf 1.29e-03 & 5.56e-02 \\ 
  5000 & 0 & 0.5 & 3 & \bf 9.25e-01 & 9.11e-01 & 8.98e-01 & 9.25e-01 & \bf 3.20e-03 & 4.40e-02 \\ 
  5000 & 0.5 & 0.5 & 1 & \bf 8.78e-01 & 8.75e-01 & 8.30e-01 & 8.76e-01 & \bf 8.48e-03 & 5.35e-02 \\ 
  5000 & 0.5 & 0.5 & 3 & \bf 9.27e-01 & 9.13e-01 & 8.91e-01 & 9.27e-01 & \bf 6.35e-03 & 3.63e-02 \\ 
  10000 & 0 & 0.5 & 1 & \bf 8.81e-01 & 8.80e-01 & 8.80e-01 & 8.78e-01 & \bf 8.70e-04 & 4.33e-02 \\ 
  10000 & 0 & 0.5 & 3 & \bf 9.25e-01 & 9.20e-01 & 8.99e-01 & 9.25e-01 & \bf 2.26e-03 & 4.26e-02 \\ 
  10000 & 0.5 & 0.5 & 1 & \bf 8.78e-01 & 8.76e-01 & 8.33e-01 & 8.76e-01 & \bf 4.72e-03 & 3.07e-02 \\ 
  10000 & 0.5 & 0.5 & 3 & \bf 9.27e-01 & 9.19e-01 & 8.93e-01 & 9.27e-01 & \bf 5.77e-03 & 3.27e-02 \\ 
  20000 & 0 & 0.5 & 1 & \bf 8.81e-01 & \bf 8.81e-01 & 8.80e-01 & 8.78e-01 & \bf 5.51e-04 & 3.09e-02 \\ 
  20000 & 0 & 0.5 & 3 & \bf 9.25e-01 & 9.21e-01 & 9.01e-01 & 9.25e-01 & \bf 1.56e-03 & 3.37e-02 \\ 
  20000 & 0.5 & 0.5 & 1 & \bf 8.78e-01 & 8.77e-01 & 8.36e-01 & 8.76e-01 & \bf 2.96e-03 & 1.56e-02 \\ 
  20000 & 0.5 & 0.5 & 3 & \bf 9.27e-01 & 9.21e-01 & 8.95e-01 & 9.27e-01 & \bf 4.57e-03 & 3.48e-02 \\ 
  30000 & 0 & 0.5 & 1 & \bf 8.81e-01 & \bf 8.81e-01 & 8.80e-01 & 8.78e-01 & \bf 2.63e-04 & 1.87e-02 \\ 
  30000 & 0 & 0.5 & 3 & \bf 9.25e-01 & 9.24e-01 & 9.02e-01 & 9.25e-01 & \bf 1.46e-03 & 2.83e-02 \\ 
  30000 & 0.5 & 0.5 & 1 & \bf 8.78e-01 & 8.77e-01 & 8.36e-01 & 8.76e-01 & \bf 2.20e-03 & 1.18e-02 \\ 
  30000 & 0.5 & 0.5 & 3 & \bf 9.27e-01 & 9.25e-01 & 8.93e-01 & 9.27e-01 & \bf 3.75e-03 & 2.98e-02 \\ 
   \hline
\end{tabular}
\caption{\tt Detailed numerical results in the binary-action setup with $\beta=0.5$. In the fifth to the eighth columns, we show the value of the best learned policy using ADR, weighted IPW, and \textit{Q-Opt} against the value of the oracle (oracle) best policy in the prespecified policy class, with all value estimates evaluated using a Monte-Carlo rollout with 20000 repeats. In the right two columns, we show the mean-squared error of the value estimates averaged across all policies in the policy class. Results are averaged across 50 runs and rounded to two decimal places. Numbers listed are accurate up to the second displaying digit due to sampling errors.} 
\label{table:setup2-beta-0_5}
\end{table}

\begin{table}[ht]
\centering
\begin{tabular}{cccccccccc}
  \hline
n & $\nu$ & $\beta$ & $\sigma$ & ADR & IPW & Q-Opt & Oracle & MSE:ADR & MSE:IPW \\ 
  \hline
250 & 0 & 1 & 1 & 7.65e-01 & 7.21e-01 & \bf 8.21e-01 & 8.03e-01 & \bf 1.56e-02 & 3.32e-01 \\ 
  250 & 0 & 1 & 3 & 8.16e-01 & 7.11e-01 & \bf 8.17e-01 & 8.46e-01 & \bf 3.20e-02 & 2.43e-01 \\ 
  250 & 0.5 & 1 & 1 & \bf 7.48e-01 & 7.25e-01 & 7.38e-01 & 7.79e-01 & \bf 1.07e-01 & 3.28e-01 \\ 
  250 & 0.5 & 1 & 3 & \bf 8.18e-01 & 7.15e-01 & 8.11e-01 & 8.49e-01 & \bf 6.97e-02 & 2.55e-01 \\ 
  500 & 0 & 1 & 1 & 7.86e-01 & 7.50e-01 & \bf 8.36e-01 & 8.03e-01 & \bf 4.70e-03 & 3.16e-01 \\ 
  500 & 0 & 1 & 3 & \bf 8.39e-01 & 7.46e-01 & 8.26e-01 & 8.46e-01 & \bf 1.27e-02 & 2.05e-01 \\ 
  500 & 0.5 & 1 & 1 & \bf 7.56e-01 & 7.28e-01 & 7.47e-01 & 7.79e-01 & \bf 7.20e-02 & 3.21e-01 \\ 
  500 & 0.5 & 1 & 3 & \bf 8.37e-01 & 7.52e-01 & 8.17e-01 & 8.49e-01 & \bf 1.98e-02 & 2.31e-01 \\ 
  1000 & 0 & 1 & 1 & 7.88e-01 & 7.74e-01 & \bf 8.46e-01 & 8.03e-01 & \bf 3.09e-03 & 2.87e-01 \\ 
  1000 & 0 & 1 & 3 & \bf 8.42e-01 & 7.81e-01 & 8.31e-01 & 8.46e-01 & \bf 9.20e-03 & 2.12e-01 \\ 
  1000 & 0.5 & 1 & 1 & \bf 7.63e-01 & 7.41e-01 & 7.49e-01 & 7.79e-01 & \bf 6.16e-02 & 2.82e-01 \\ 
  1000 & 0.5 & 1 & 3 & \bf 8.42e-01 & 7.81e-01 & 8.16e-01 & 8.49e-01 & \bf 2.87e-02 & 2.03e-01 \\ 
  5000 & 0 & 1 & 1 & 7.97e-01 & 7.94e-01 & \bf 8.54e-01 & 8.03e-01 & \bf 1.62e-03 & 2.05e-01 \\ 
  5000 & 0 & 1 & 3 & \bf 8.46e-01 & 8.32e-01 & 8.39e-01 & 8.46e-01 & \bf 4.92e-03 & 1.42e-01 \\ 
  5000 & 0.5 & 1 & 1 & \bf 7.68e-01 & 7.60e-01 & 7.53e-01 & 7.79e-01 & \bf 2.46e-02 & 1.58e-01 \\ 
  5000 & 0.5 & 1 & 3 & \bf 8.49e-01 & 8.19e-01 & 8.22e-01 & 8.49e-01 & \bf 1.37e-02 & 1.39e-01 \\ 
  10000 & 0 & 1 & 1 & 8.00e-01 & 7.98e-01 & \bf 8.54e-01 & 8.03e-01 & \bf 9.82e-04 & 1.69e-01 \\ 
  10000 & 0 & 1 & 3 & \bf 8.46e-01 & 8.40e-01 & 8.42e-01 & 8.46e-01 & \bf 4.53e-03 & 1.52e-01 \\ 
  10000 & 0.5 & 1 & 1 & \bf 7.69e-01 & 7.64e-01 & 7.52e-01 & 7.79e-01 & \bf 1.45e-02 & 9.67e-02 \\ 
  10000 & 0.5 & 1 & 3 & \bf 8.49e-01 & 8.33e-01 & 8.23e-01 & 8.49e-01 & \bf 1.10e-02 & 1.53e-01 \\ 
  20000 & 0 & 1 & 1 & 8.00e-01 & 8.01e-01 & \bf 8.54e-01 & 8.03e-01 & \bf 3.54e-04 & 9.76e-02 \\ 
  20000 & 0 & 1 & 3 & \bf 8.46e-01 & 8.42e-01 & 8.42e-01 & 8.46e-01 & \bf 3.23e-03 & 1.26e-01 \\ 
  20000 & 0.5 & 1 & 1 & \bf 7.69e-01 & 7.67e-01 & 7.54e-01 & 7.79e-01 & \bf 7.74e-03 & 5.15e-02 \\ 
  20000 & 0.5 & 1 & 3 & \bf 8.49e-01 & 8.42e-01 & 8.21e-01 & 8.49e-01 & \bf 1.01e-02 & 1.42e-01 \\ 
  30000 & 0 & 1 & 1 & 8.00e-01 & 8.00e-01 & \bf 8.54e-01 & 8.03e-01 & \bf 3.22e-04 & 9.42e-02 \\ 
  30000 & 0 & 1 & 3 & \bf 8.46e-01 & 8.43e-01 & 8.41e-01 & 8.46e-01 & \bf 2.85e-03 & 1.19e-01 \\ 
  30000 & 0.5 & 1 & 1 & \bf 7.70e-01 & 7.68e-01 & 7.54e-01 & 7.79e-01 & \bf 7.99e-03 & 5.27e-02 \\ 
  30000 & 0.5 & 1 & 3 & \bf 8.49e-01 & 8.44e-01 & 8.22e-01 & 8.49e-01 & \bf 9.88e-03 & 1.11e-01 \\ 
   \hline
\end{tabular}
\caption{\tt Detailed numerical results in the binary-action setup with $\beta=1$; details see caption of Table \ref{table:setup2-beta-0_5}.} 
\label{table:setup2-beta-1}
\end{table}

\begin{table}[ht]
\centering
\begin{tabular}{cccccccc}
  \hline
n & $\sigma$ & ADR & IPW & Q-Opt & Oracle & MSE:ADR & MSE:IPW \\ 
  \hline
250 & 0 & \bf 1.65e-01 & 1.40e-01 & 5.88e-02 & 2.65e-01 & \bf 1.79e-01 & 1.35e+00 \\ 
  250 & 0.5 & \bf 1.24e-01 & 9.13e-02 & 5.85e-02 & 2.67e-01 & \bf 1.83e-01 & 1.26e+00 \\ 
  250 & 1 & \bf 1.17e-01 & 8.94e-02 & 8.67e-02 & 2.54e-01 & \bf 1.97e-01 & 1.38e+00 \\ 
  500 & 0 & \bf 1.88e-01 & 1.14e-01 & 1.08e-01 & 2.65e-01 & \bf 8.74e-02 & 6.39e-01 \\ 
  500 & 0.5 & \bf 1.72e-01 & 1.13e-01 & 1.00e-01 & 2.67e-01 & \bf 8.30e-02 & 6.55e-01 \\ 
  500 & 1 & \bf 1.48e-01 & 9.30e-02 & 1.04e-01 & 2.54e-01 & \bf 9.00e-02 & 6.12e-01 \\ 
  1000 & 0 & \bf 1.98e-01 & 1.52e-01 & 1.38e-01 & 2.65e-01 & \bf 3.55e-02 & 2.78e-01 \\ 
  1000 & 0.5 & \bf 1.99e-01 & 1.40e-01 & 1.15e-01 & 2.67e-01 & \bf 4.28e-02 & 3.28e-01 \\ 
  1000 & 1 & \bf 1.83e-01 & 1.54e-01 & 1.38e-01 & 2.54e-01 & \bf 4.26e-02 & 3.65e-01 \\ 
  5000 & 0 & \bf 2.35e-01 & 1.94e-01 & 2.23e-01 & 2.65e-01 & \bf 7.04e-03 & 5.99e-02 \\ 
  5000 & 0.5 & \bf 2.21e-01 & 2.09e-01 & 2.02e-01 & 2.67e-01 & \bf 7.47e-03 & 6.36e-02 \\ 
  5000 & 1 & \bf 2.26e-01 & 1.75e-01 & 1.97e-01 & 2.54e-01 & \bf 8.21e-03 & 6.43e-02 \\ 
  10000 & 0 & \bf 2.49e-01 & 2.19e-01 & 2.43e-01 & 2.65e-01 & \bf 3.36e-03 & 3.04e-02 \\ 
  10000 & 0.5 & \bf 2.36e-01 & 2.06e-01 & 2.22e-01 & 2.67e-01 & \bf 3.95e-03 & 3.27e-02 \\ 
  10000 & 1 & \bf 2.27e-01 & 1.88e-01 & 2.15e-01 & 2.54e-01 & \bf 4.68e-03 & 3.14e-02 \\ 
  20000 & 0 & 2.47e-01 & 2.20e-01 & \bf 2.62e-01 & 2.65e-01 & \bf 1.75e-03 & 1.51e-02 \\ 
  20000 & 0.5 & \bf 2.45e-01 & 2.08e-01 & 2.39e-01 & 2.67e-01 & \bf 2.02e-03 & 1.47e-02 \\ 
  20000 & 1 & \bf 2.35e-01 & 2.22e-01 & \bf 2.35e-01 & 2.54e-01 & \bf 2.92e-03 & 1.46e-02 \\ 
   \hline
\end{tabular}
\caption{\tt Detailed numerical results in the multiple-action setup. In the third to the sixth columns, we show the value of the best learned policy using ADR, IPW, and \textit{Q-Opt} against the value of the oracle best policy in the prespecified policy class, with all value estimates evaluated using a Monte-Carlo rollout with 20000 repeats. In the right two columns, we show the mean-squared error of the value estimates averaged across all policies in the policy class. Results are averaged across 50 runs and rounded to two decimal places. Numbers listed are accurate up to the second displaying digit due to sampling errors.} 
\label{table:setup1}
\end{table}

\end{appendix}
\end{document}